\documentclass[letterpaper,11pt]{article}
\usepackage[utf8]{inputenc}

\pdfoutput=1

\usepackage{amsmath, amsthm, amssymb, thm-restate}
\usepackage{algorithmicx}
\usepackage[table,xcdraw]{xcolor}
\usepackage[vlined,linesnumbered]{algorithm2e}
\usepackage{setspace}
\usepackage{mathtools}
\usepackage[numbers]{natbib}
\usepackage{comment} 
\usepackage[breakable]{tcolorbox} 
\usepackage{xfrac}
\usepackage{hyperref}
\usepackage{multirow}
\usepackage{caption}
\usepackage{bm}
\usepackage{newfloat}
\usepackage{enumitem}
\usepackage{dblfloatfix} 
\usepackage{wrapfig}

\tcbuselibrary{skins,breakable}
\tcbset{enhanced jigsaw}

\usepackage[margin=1in]{geometry}

\allowdisplaybreaks

\definecolor{myblue}{RGB}{10,110,230}
\definecolor{mygreen}{RGB}{10,150,110}
\definecolor{myred}{RGB}{170,40,40}

\hypersetup{
     colorlinks=true,
     citecolor= mygreen,
     linkcolor= myred
}

\renewcommand{\epsilon}{\varepsilon}

\DeclareMathOperator{\E}{\ensuremath{\normalfont \textbf{E}}}

\newcommand{\hiddencomment}[1]{}

\newcommand{\mc}[1]{\ensuremath{\mathcal{#1}}}

\newcommand{\GMM}[1]{\ensuremath{\mathsf{GMM}(#1)}}

\newcommand{\finalapxgl}[0]{\ensuremath{.501}}

\DeclareMathOperator{\poly}{poly}

\usepackage[noabbrev,nameinlink]{cleveref}
\crefname{lemma}{Lemma}{Lemmas}
\crefname{theorem}{Theorem}{Theorems}
\crefname{property}{Property}{Properties}
\crefname{claim}{Claim}{Claims}
\crefname{definition}{Definition}{Definitions}
\crefname{observation}{Observation}{Observations}
\crefname{proposition}{Proposition}{Propositions}
\crefname{assumption}{Assumption}{Assumptions}
\crefname{line}{Line}{Lines}
\crefname{figure}{Figure}{Figures}
\creflabelformat{property}{(#1)#2#3}
\crefname{equation}{}{}
\crefname{section}{Section}{Sections}
\crefname{appendix}{Appendix}{Appendices}
\crefname{openproblem}{Open Problem}{Open Problems}
\crefname{algCounter}{Algorithm}{Algorithms}
\Crefname{algCounter}{Algorithm}{Algorithms}

\newtheorem{openproblem}{Open Problem}
\newtheorem{theorem}{Theorem}
\newtheorem{lemma}{Lemma}[section]
\newtheorem{proposition}[lemma]{Proposition}

\newtheorem{definition}[lemma]{Definition}
\newtheorem{claim}[lemma]{Claim}
\newtheorem{fact}[lemma]{Fact}
\newtheorem{observation}[lemma]{Observation}
\newtheorem{remark}[lemma]{Remark}

\definecolor{mylightgray}{RGB}{230,230,230}

\newcommand{\smparagraph}[1]{\vspace{-0.2cm}\paragraph{#1}}

\algnewcommand{\IIf}[2]{\textbf{if} #1 \textbf{then} #2}
\algnewcommand{\EndIIf}{\unskip\ \algorithmicend\ \algorithmicif}

\newenvironment{graytbox}{
\par\addvspace{0.1cm}
\begin{tcolorbox}[width=\textwidth,
                  boxsep=5pt,
                  left=1pt,
                  right=1pt,
                  top=2pt,
                  bottom=2pt,
                  boxrule=0pt,
                  arc=0pt,
                  colback=mylightgray,
                  colframe=black,
                  ]
}{
\end{tcolorbox}
}

\newenvironment{whitetbox}{
\par\addvspace{0.1cm}
\begin{tcolorbox}[width=\textwidth,
                  boxsep=5pt,
                  left=1pt,
                  right=1pt,
                  top=2pt,
                  bottom=2pt,
                  boxrule=1pt,
                  arc=0pt,
                  colframe=black,
                  colback=white,
                  breakable
                  ]
}{
\end{tcolorbox}
}

\newenvironment{myproof}{
\vspace{-0.5cm}
\begin{proof}
}{
\end{proof}
}

\newcounter{algCounter}

\newenvironment{alglist}[2]{
    \begin{whitetbox}
        \refstepcounter{algCounter}
        \textbf{Algorithm~\thealgCounter:} #1
        \label{#2}
        
        \vspace{-0.2cm}
}{
    \end{whitetbox}
}

\makeatletter
\renewcommand{\paragraph}{%
  \@startsection{paragraph}{4}%
  {\z@}{10pt}{-1em}%
  {\normalfont\normalsize\bfseries}%
}
\makeatother

\makeatletter
\patchcmd{\@algocf@start}
  {-1.5em}
  {0pt}
  {}{}
\makeatother

\title{Dynamic Algorithms for Maximum Matching Size}

\author{
Soheil Behnezhad\thanks{Khoury College of Computer Sciences, Northeastern University. Webpage: \texttt{behnezhad.com}.}
}

 \date{}

\begin{document}

\maketitle

\thispagestyle{empty}

\begin{abstract}
    We study fully dynamic algorithms for maximum matching. This is a well-studied problem, known to admit several update-time/approximation trade-offs. For instance, it is known how to maintain a 1/2-approximate matching in $(\poly\log n)$ update time or a $2/3$-approximate matching in $O(\sqrt{n})$ update time, where $n$ is the number of vertices. It has been a long-standing open problem to determine whether either of these bounds can be improved.
    
    \smallskip\smallskip
    
    In this paper, we show that when the goal is to maintain just the size of the matching (and not its edge-set), then these bounds can indeed be improved. First, we give an algorithm that takes $(\poly\log n)$ update-time and maintains a $.501$-approximation ($.585$-approximation if the graph is bipartite). Second, we give an algorithm that maintains a $(2/3 + \Omega(1))$-approximation in $O(\sqrt{n})$ time for bipartite graphs.
    
    \smallskip\smallskip
    
    Our results build on new connections to sublinear time algorithms. In particular, a key tool for both is an algorithm of the author for estimating the size of maximal matchings in $\widetilde{O}(n)$ time [Behnezhad; FOCS 2021]. Our second result also builds on the {\em edge-degree constrained subgraph (EDCS)} of Bernstein and Stein [ICALP'15, SODA'16]. In particular, while it has been known that EDCS may not include a better than 2/3-approximation, we give a new characterization of such tight instances which allows us to break it. We believe this characterization might be of independent interest.
\end{abstract}

\clearpage

{
\hypersetup{hidelinks}
\vspace{1cm}
\renewcommand{\baselinestretch}{0.1}
\setcounter{tocdepth}{2}
\tableofcontents{}
\thispagestyle{empty}
\clearpage
}

\setcounter{page}{1}

\clearpage
\section{Introduction}

We study approximate maximum matchings in fully dynamic graphs. We are given an $n$-vertex graph $G=(V, E)$ that is subject to both edge insertions and deletions. The goal is to maintain (the size of) an approximate maximum matching of $G$ while spending a small time per update.  For exact maximum matchings, there are conditional lower bounds that rule out any $O(n^{1-\epsilon})$ update-time algorithm \cite{AbboudW14,HenzingerKNS15,Dahlgaard-ICALP16} (see also \cite{Sankowski07,BrandNS19} for upper bounds). As such, much of the work in the literature of dynamic matching has been on approximate solutions; see \cite{OnakR10,BaswanaGS11,BaswanaGS18,NeimanS-STOC13,GuptaPeng-FOCS13,BhattacharyaHI-SiamJC18,BernsteinS-ICALP15,BernsteinSteinSODA16,BhattacharyaHN-SODA17,BhattacharyaHN-STOC16,Solomon-FOCS16,CharikarS18-ICALP,ArarCCSW-ICALP18,BernsteinFH-SODA19,BehnezhadDHSS-FOCS19,BehnezhadLM-SODA20,Wajc-STOC20,BernsteinDL-STOC21,BhattacharyaK21-ICALP21,RoghaniSW-ITCS22,Kiss-ITCS22,GrandoniSSU-SOSA22,BehnezhadK-SODA22} and their references.

\paragraph{$1/2$-Approximation:} When the goal is to maintain a $1/2$-approximation, then extremely fast $(\poly\log n)$ update time algorithms have been known since the work of \citet*{BaswanaGS11} in 2011 (see also \cite{Solomon-FOCS16,BernsteinFH-SODA19,BehnezhadDHSS-FOCS19}). These algorithms maintain a greedy maximal matching, for which the 1/2-approximation guarantee is tight. In sharp contrast, all known $(\frac{1}{2}+\Omega(1))$-approximate algorithms require a polynomial update-time of $n^{\Omega(1)}$ \cite{BhattacharyaHN-STOC16,BehnezhadLM-SODA20,Wajc-STOC20,BehnezhadK-SODA22}. The following question, in particular, has been a major open problem of the area for more than a decade:

\begin{openproblem}\label{open:half}
    Is it possible to maintain a $(\frac{1}{2} + \Omega(1))$-approximation in $\poly\log n$ update-time?
\end{openproblem}

To our knowledge, \Cref{open:half} was first asked by \citet*{OnakR10} in 2010. Several subsequent papers have also imposed it as an important open problem. See e.g. \cite[Section~4]{BhattacharyaHN-STOC16}, \cite[Section~7]{BernsteinSteinSODA16}, \cite[Section~1]{CharikarS18-ICALP}, \cite[Section~1]{BehnezhadLM-SODA20}, \cite[Section~5]{Wajc-STOC20}, and \cite[Section~5]{BehnezhadK-SODA22}.

The fastest known better-than-1/2 approximation requires $O(n^\epsilon)$ update time and obtains a $(\frac{1}{2}+f(\epsilon))$-approximation, where $\epsilon > 0$ can be any constant. This was proved for general graphs by \citet*{BehnezhadLM-SODA20}. Prior to that, an algorithm with the same trade-off was proposed by \citet*{BhattacharyaHN-STOC16} for maintaining the size (but not the edges) of the matching in bipartite graphs (see also \cite{Wajc-STOC20,BhattacharyaK21-ICALP21,BehnezhadK-SODA22}).

Our first main result is a positive resolution of \cref{open:half} for maintaining the size (but not the edges) of the matching.

\begin{graytbox}
\begin{theorem}\label{thm:main}
    There is a fully dynamic algorithm that maintains a $\finalapxgl$-approximation of the size of maximum matching in $(\poly\log n)$ worst-case update time. The algorithm is randomized but works against adaptive adversaries.
\end{theorem}
\end{graytbox}

As a warm-up to our techniques, we prove an equivalent of \cref{thm:main} for bipartite graphs in \cref{sec:bipartite} (see \cref{thm:main-bipartite}). The approximation ratio turns out to be much better for bipartite graphs: we get a $(2-\sqrt{2}) \approx .585$-approximation against oblivious adversaries, and a $.542$-approximation against adaptive adversaries.

\paragraph{$2/3$-Approximation:} Another notable approximation/update-time trade-off was obtained by \citet*{BernsteinS-ICALP15,BernsteinSteinSODA16} in 2015, who showed that a $(2/3-\epsilon)$-approximation can be maintained in $O(\sqrt{n} \cdot \poly(1/\epsilon))$ update time for bipartite graphs (see also \cite{BernsteinSteinSODA16,GrandoniSSU-SOSA22,Kiss-ITCS22,BehnezhadK-SODA22} for generalizations of this result to non-bipartite graphs). They achieved this by introducing an elegant matching sparsifier that they called an {\em edge-degree constrained subgraph} (EDCS). It is known that the guarantee of 2/3-approximation is tight for EDCS. That is, there are inputs on which the algorithm of \citet*{BernsteinSteinSODA16} does not obtain a better-than-2/3 approximation.  Once we go slightly above 2/3-approximation, then the fastest known algorithms require a much larger nearly-linear-in-$n$ update-time \cite{GuptaPeng-FOCS13}. In particular, a recent algorithm of \citet*{AssadiBKL-Arxiv22} can be used to maintain a $(1-o(1))$-approximate matching in $n/(\log^* n)^{\Theta(1)}$ update-time.


It is worth noting that the 2/3-approximation has been a barrier in several settings. Most notably, is the {\em two-party one-way communication} model. For that model, \citet*{AssadiBernstein19} showed that an (almost) $2/3$-approximation can be achieved with $O(n)$ communication. On the flip side, \citet*{GoelKK12} proved that a $(2/3+\Omega(1))$-approximation requires a much larger $n^{1+\Omega(1/\log\log n)} \gg n \poly\log n$ communication. See also the paper of \citet*{AssadiBehnezhad21} on the importance of the 2/3-approximation barrier, and discussions about it in the random-order streaming model. This motivates the next open question, see e.g. \cite[Section~5]{BehnezhadK-SODA22}.

\begin{openproblem}\label{open:2/3}
    Is it possible to maintain a $(\frac{2}{3} + \Omega(1))$-approximation in $O(\sqrt{n})$ update-time?
\end{openproblem}


Our second main result is a positive resolution of \cref{open:2/3} for bipartite graphs, also for the matching size.

\begin{graytbox}
\begin{theorem}\label{thm:2/3}
    For an absolute constant $\delta_0 > 10^{-6}$, there is an algorithm that maintains a $(\frac{2}{3}+\delta_0)$-approximation of the size of maximum matching in $O(\min\{m^{1/4}, \sqrt{\Delta}\} + \poly\log n) = O(\sqrt{n})$ worst-case update time. Here $\Delta$ is a fixed upper bound on the maximum degree and $m$ is the number of edges. The algorithm is randomized but works against adaptive adversaries.
\end{theorem}
\end{graytbox}

To achieve \cref{thm:2/3}, we prove a new characterization of tight instances of EDCS that might be of independent interest given the versatility of EDCS. We briefly overview this in \cref{sec:techniques}. See \cref{sec:EDC-char-statement} for the formal statement of the characterization and a comparison with prior work.

\paragraph{Concurrent Work:} In an independent and concurrent work, and similar to our \cref{thm:main}, \citet*{BhattacharyaKSW-SODA23} show that a better-than-$1/2$ approximation of maximum matching size can be maintained in $(\poly\log n)$ update time, also answering \cref{open:half} in the affirmative. The high-level approach of both works is the same. In particular, the key new ingredient in both is to use the sublinear time algorithm of \citet*{behnezhad2021} for augmenting a 1/2-approximate matching. However, the details are different (in particular under adaptive adversaries, \cite{BhattacharyaKSW-SODA23} obtains a better quantitative improvement over 1/2). We note that \cref{thm:2/3} for beating $2/3$-approximations is unique to our paper.


\vspace{-0.4cm}
\subsection{Perspective: On Maintaining Size vs. Edges}

While majority of the algorithms in the literature maintain the edge-set of the matching, dynamic algorithms for maintaining the matching size have also been studied in several works. For instance, progress on dynamic algorithms for exact maximum matchings has been merely made on maintaining its size \cite{Sankowski07,BrandNS19}. These algorithms use algebraic techniques. Maintaining the size has also been studied for approximate matchings. In particular, \citet*{BhattacharyaHN-STOC16} studied dynamic algorithms for $(\frac{1}{2} + \Omega(1))$-approximating the maximum matching size in bipartite graphs. The algorithm of \cite{BhattacharyaHN-STOC16} maintains a fractional matching, but not an integral one. The bound on the size of the maximum integral matching, nonetheless, follows from the fact that fractional matchings have integrality gap 1 for bipartite graphs. We note that after the work of \cite{BhattacharyaHN-STOC16}, algorithms have been developed for ``losslessly'' rounding fractional matchings into integral ones dynamically \cite{Wajc-STOC20,BhattacharyaK21-ICALP21}. As such, fractional matchings can no longer be a source of separation for maintaining the size vs. the edge-set.

The fact that our algorithms in \cref{thm:main,thm:2/3} estimate the size of the maximum matching as opposed to its edge-set is due to a very different reason than those mentioned above. One of the main building blocks of both algorithms  is a recent {\em sublinear time} algorithm of the author \cite{behnezhad2021} that 1/2-approximates the maximum matching size in $\widetilde{O}(n)$ time.\footnote{Here and throughout the paper, $\widetilde{O}(f) = f \cdot \poly\log n$.} Had the algorithm in \cite{behnezhad2021} worked for {\em finding} a 1/2-approximate matching instead of just its size, then our algorithms could have also maintained the edges of the matching. But this is known to be impossible: $\Omega(n^2)$ queries are information theoretically needed to find any $O(1)$-approximate maximum matching in the adjacency matrix model, even when the graph is promised to have a perfect matching \cite{behnezhad2021}.
\footnote{The lower bound is simple. Suppose that the graph is a perfect matching with its vertex IDs being randomly permuted. Finding this matching requires $\Omega(n^2)$ queries to the adjacency matrix.}

In light of the new connection to sublinear time algorithms discovered in this work and the existing separation for approximating the size vs. finding the edge-set of the matching for sublinear time algorithms, it remains a tantalizing future direction to explore whether such separations exist in the dynamic model. Conversely, resolving either of \cref{open:half,open:2/3} positively for maintaining the edge-set of the matching would be an excellent result.

\vspace{-0.4cm}

\section{Technical Overview}\label{sec:techniques}

Here we give an overview of our algorithms for \cref{thm:main,thm:2/3}. To convey the main intuitions in this section, let us make the simplifying assumption that the maximum matching size $\mu(G)$ of the graph $G$ remains $\Omega(n)$ at all times. We note that this assumption can easily be lifted with just a $(\poly\log n)$ overhead in the update-time using a randomized ``vertex sparsification'' idea of the literature \cite{AssadiKL17} (which can be made work against adaptive adversaries too \cite{Kiss-ITCS22}).

\paragraph{A Reduction to Sublinear Algorithms:} Suppose that we have an algorithm $\mc{A}$ that in $T$ time $\alpha$-approximates the maximum matching size of a static $n$-vertex graph. This can be turned into a dynamic algorithm with the ``lazy'' approach: run $\mc{A}$, keep the solution unchanged for $\Theta(\epsilon n)$ updates, then repeat. Since the maximum matching size changes by at most one in each update and $\mu(G) = \Omega(n)$, then we have an $(\alpha-\epsilon)$-approximation at all times. The (amortized) update-time of the algorithm, on the other hand, is only $O(T/\epsilon n)$. How is this useful? There is a rich body of work on sublinear time algorithms for maximum matching. Building on a long and beautiful line of work \cite{ParnasRon07,YoshidaYISTOC09,NguyenOnakFOCS08,OnakSODA12,ChenICALP20,Kapralov21,behnezhad2021}, the author showed that a $(1/2-\epsilon)$-approximation can be obtained only in $\widetilde{O}(n/\epsilon^2)$ time \cite{behnezhad2021}, which is much smaller than the number of edges in the graph that can be as large as $\Omega(n^2)$. Plugged into the framework above, this gives an (almost) 1/2-approximation in $(\poly\log n)$ update-time. Up until very recently, no $o(n^2)$ time algorithm was known for beating 1/2-approximation. This long-standing barrier was very recently broken by \citet*{BehnezhadRRS-SODA23}, who gave an $O(n^{1+\epsilon})$ time algorithm that obtains a $(\frac{1}{2}+\Omega_\epsilon(1))$-approximation. Thus, the same approximation can be maintained in $\widetilde{O}(n^\epsilon)$ update-time using the framework above. Unfortunately, however, this is slower than our desired $(\poly\log n)$ time, and somewhat curiously, matches the guarantee of the previous algorithms for maintaining the {\em edges} of the matching \cite{BhattacharyaHN-STOC16,BehnezhadLM-SODA20}.

\paragraph{Beyond Sublinear Algorithms:} For traditional sublinear time algorithms, the only input provided is the graph itself. That is, the algorithm is only given query access to the graph (either to its adjacency list or its adjacency matrix). This, however, does not need to be the case for our target application in the framework above. Let us explain this with an example. As discussed, there are algorithms that take $\poly\log n$ worst-case update-time, and maintain the edges of a maximal matching of a fully dynamic graph \cite{BaswanaGS18,Solomon-FOCS16,BernsteinFH-SODA19,BehnezhadDHSS-FOCS19}. One thing we can do, is to run such dynamic algorithms in the background. This way, when the time comes to call the static sublinear time algorithm to estimate the maximum matching size, in addition to providing query access to the graph, we can also feed this maximal matching into the sublinear time algorithm for free. This is, in fact, exactly what we do to prove \cref{thm:main}. In essence, we show that provided adjacency matrix access and given the edge-set of a maximal matching, there is a sublinear time algorithm that $(\frac{1}{2}+\Omega(1))$-approximates the maximum matching size in $\widetilde{O}(n)$ time. Plugged into the framework above, this leads to a $(\frac{1}{2} + \Omega(1))$-approximate fully dynamic algorithm with $\poly\log n$ update-time. For bipartite graphs, this sublinear time algorithm turns out to be very simple and clean, and achieves a much better approximation than half. We prove this as a warm-up in \cref{sec:bipartite}.

\paragraph{Beating 2/3:} For our \cref{thm:2/3} which beats 2/3-approximation, we follow the same framework discussed above. However, instead of feeding a maximal matching into the sublinear algorithm, we maintain an edge-degree constrained subgraph (EDCS) and feed that into the sublinear algorithm. The EDCS, introduced by \citet*{BernsteinS-ICALP15}, is
a sparse subgraph that, for the right parameters, includes a 2/3-approximate maximum matching of its base graph. It is known that this bound is tight --- that there are graphs for which an EDCS does not include a better than 2/3-approximation. To beat 2/3-approximation, one natural idea is to find the 2/3-approximate matching inside the EDCS, and then try to augment it. In general, a 2/3-approximate matching may not leave any augmenting path of length shorter than 5. It seems challenging to find (or even estimate the number of) length-5 augmenting paths efficiently in our model. To get around this, we prove a new characterization of the tight instances of EDCS. We show that when an EDCS $H$ does not include a strictly larger than $2/3$-approximate matching of its base graph $G$, then it must include a $2/3$-approximate matching $M$ that is far from being maximal. In fact, we show that there is a $1/3$-approximate matching in $G$ whose edges can be directly added to $M$. We believe this property might be of independent interest given the versatility of EDCS. See \cref{sec:EDC-char-statement} for more about this and a comparison to a previous characterization by Assadi and the author \cite{AssadiBehnezhad21}.

\vspace{-0.4cm}
\section{Preliminaries}
\vspace{-0.2cm}

Given a graph $G=(V, E)$ and a subset $U \subseteq V$, we use $G[U]$ to denote the induced subgraph of $G$ on $U$. Given two disjoint subsets $A$, $B$ of $V$, we use $G[A, B]$ to denote the bipartite subgraph of $G$ including any of its edges between $A$ and $B$. Given a vertex $v$, we use $N_G(v)$ to denote the set of neighbors of $v$ in $G$. We may drop the subscript when the graph $G$ is clear from the context.

A matching $M$ for $G$ is a subset of its edges such that no two of them share an endpoint. A maximum matching in $G$ is a matching of the largest possible size. We use $\mu(G)$ to denote the size of the maximum matching in $G$. We say a number $\widetilde{\mu}(G)$ is an $\alpha$-approximation of $\mu(G)$ for $\alpha \in (0, 1]$ if $\alpha \mu(G) \leq \widetilde{\mu}(G) \leq \mu(G)$. 

Given a matching $M$, we use $V(M)$ to denote the set of vertices in $M$. For simplicity, given a vertex $v$ we may write $v \in M$ instead of $v \in V(M)$. Given two matchings $M, M^\star$, we use $M \oplus M^\star$ to denote the symmetric difference of the two matchings, i.e., the edges that belong to exactly one of $M$ and $M^\star$. Given a permutation $\pi$ over the edge-set $E$ of a graph $G$, we use $\GMM{G, \pi}$ to denote the matching obtained by greedily processing the edges of $G$ in the order of $E$ and adding each edge possible to the matching.

Throghout the paper {\em w.h.p.} abbreviates {\em with high probability} by which we mean probability at least $1-1/n^c$ for any desirably large constant $c \geq 1$.

\section{Warm-Up: Beating Half for Bipartite Graphs}\label{sec:bipartite}

In this section, we prove the following theorem for bipartite graphs.\footnote{We note that in the first version of the paper, we obtained a $.534$-approximation for bipartite graphs. We thank an anonymous SODA'23 reviewer who suggested a tweak in the analysis, leading to the improved bound of \cref{thm:main-bipartite}.}

\begin{theorem}\label{thm:main-bipartite}
    For bipartite graphs, there is a randomized fully dynamic algorithm that maintains a $(1-\epsilon)(2-\sqrt{2}) \approx .585$-approximation of the size of maximum matching in $(\poly\log n)$ worst-case update-time against oblivious adversaries. If the adversary is adaptive, then there is an algorithm that maintains a $.542$ approximation in $(\poly\log n)$ worst-case update-time.
\end{theorem}

First, we start with the algorithm for oblivious adversaries. Then show how to lift the oblivious adversary assumption.

Let us first state a general lemma that we use in all of our results. The lemma guarantees that given a ``semi-dynamic'' algorithm $\mc{A}$ (which only upon being queried produces an estimate of the matching size), one can turn it into a dynamic algorithm $\mc{B}$ for maintaining the maximum matching size at all times.

\newcommand{\lemAtoBstatement}[0]{Let $G$ be an $n$-vertex fully dynamic graph and let $\epsilon > 0$ be a parameter. Suppose that there is a (randomized) data structure $\mc{A}$ (this is the semi-dynamic algorithm) that takes $U(n)$ worst-case time per update to $G$ and, upon being queried, $\mc{A}$ produces in $Q(n, \epsilon)$ time a number $\widetilde{\mu}$, such that $\alpha \mu(G) - \epsilon n \leq \E[\widetilde{\mu}] \leq \mu(G)$. Then there is a randomized  data structure $\mc{B}$ (this is the fully-dynamic algorithm) that maintains a number $\widetilde{\mu}'$ such that at any point during the updates, w.h.p., $(\alpha-\epsilon)\mu(G) \leq \widetilde{\mu}' \leq \mu(G)$. Algorithm $\mc{B}$ takes $O\left(\left(U(n) + \frac{Q(n, \epsilon^2)}{n} \right) \poly(\log n, 1/\epsilon) \right)$ worst-case update-time. Moreover, if $\mc{A}$ works against adaptive adversaries, so does  $\mc{B}$.}
\begin{lemma}[\textbf{From Semi-Dynamic Algorithms to Fully-Dynamic Algorithms}]\label{lem:A-to-B}
\lemAtoBstatement{}
\end{lemma}

We emphasize that we do not claim any novelty for \cref{lem:A-to-B}; its proof, which we present in \cref{sec:A-to-B-proof}, stitches together known ideas from the literature. Our main novelty for all of our results, is to provide the semi-dynamic algorithm that we plug into \cref{lem:A-to-B}. As a warm-up in this section, we describe this semi-dynamic algorithm for beating half-approximation in bipartite graphs. In particular, we prove the following lemma.

\begin{lemma}\label{lem:A-for-bipartite}
    For any $\epsilon > 0$ and any $n$-vertex fully dynamic graph $G$, there is a (randomized) data structure $\mc{A}$ that takes $(\poly\log n)$ worst-case update-time against an oblivious adversary, and upon being queried takes $\widetilde{O}(n/\epsilon^3)$ time to produce a number $\widetilde{\mu}$ such that 
    $
        (2-\sqrt{2}) \cdot \mu(G) - \epsilon n \leq \E[\widetilde{\mu}] \leq \mu(G).
    $
\end{lemma}

\cref{lem:A-to-B} and \cref{lem:A-for-bipartite} together prove \cref{thm:main-bipartite} against oblivious adversaries.

\begin{proof}[Proof of \cref{thm:main} against oblivious adversaries]
    Follows by plugging the semi-dynamic algorithm given by \cref{lem:A-for-bipartite} into \cref{lem:A-to-B}.
\end{proof}

\subsection{The Semi-Dynamic Algorithm (Proof of \cref{lem:A-for-bipartite})}

We start by describing the data structures that our semi-dynamic algorithm maintains. \clearpage

\smparagraph{Data Structures:} We maintain the \underline{adjacency matrix} of the graph $G$.\footnote{Storing the adjacency matrix naively requires $O(n^2)$ space. While space-complexity is often not a concern for dynamic algorithms, we note that the space can also be reduced to $\widetilde{O}(m)$ by simply storing the non-zero entries of the adjacency matrix in a binary search tree. This only blows up the time-complexity by a $O(\log n)$ factor.} Additionally, we maintain a maximal matching $M$ of $G$ which can be done in $(\poly\log n)$ worst-case update-time against oblivious adversaries \cite{BernsteinFH-SODA19,BehnezhadDHSS-FOCS19} (this is the only part of the algorithm that requires the oblivious adversary assumption, we show how this can be lifted at the end of this section). Therefore, overall, the dynamic algorithm requires $(\poly\log n)$ worst-case update-time.

It remains to show how to produce the number $\widetilde{\mu}$ in $\widetilde{O}(n/\epsilon^3)$ time using these data structures, which is what we focus on in the rest of this section.

\smparagraph{The Query Algorithm:} Our starting point is the following \cref{alg:static}, which is inspired by a random-order streaming algorithm of \citet*{KonradMM12}.

\begin{alglist}{}{alg:static}
\begin{enumerate}[leftmargin=15pt, itemsep=0pt]
    \item Let $M$ be the maximal matching of $G$ that we maintain.
    
    \item Let $M' \subseteq M$ include each edge of $M$ independently with probability $p = \sqrt{2}-1$.
    
    \item Let $V' := V(M')$ be the set of vertices matched by $M'$, and let $U := V \setminus V(M)$ be the vertices left unmatched by maximal matching $M$.
    
    \item Let $H := G[V', U]$ be the induced bipartite subgraph of $G$ between $V'$ and $U$. (We will not construct this graph $H$ explicitly.)
    
    \item Let $g := \E_\pi[|\GMM{H, \pi}| \mid M']$.
    
    \item Return $\widetilde{\mu}' := |M| + \max\{0, g - |M'|\}$.
\end{enumerate}
\end{alglist}

 Implementing \cref{alg:static}, as stated, requires $\Omega(n^2)$ time which is much higher than our desired time in \cref{lem:A-for-bipartite}. However, we show that one can estimate its output up to an additive $\epsilon n$ error (which note can be tolerated by \cref{lem:A-for-bipartite}) much faster  in $\widetilde{O}(n/\epsilon^3)$ time (\cref{lem:implementation}). This quick implementation of \cref{alg:static} is what we use when the semi-dynamic algorithm is queried.

\begin{lemma}\label{lem:implementation}
    For any $\epsilon > 0$, there is an algorithm that w.h.p. takes $\widetilde{O}(n/\epsilon^3)$ time and returns a number $\widetilde{\mu}''$ such that $\widetilde{\mu}' - \epsilon n \leq \widetilde{\mu}'' \leq \widetilde{\mu}'$. Here $\widetilde{\mu}'$ is the output of \cref{alg:static}.
\end{lemma}

To prove \cref{lem:implementation}, we use the following sublinear-time greedy matching estimator of the author \cite{behnezhad2021}. See \cref{sec:sublinear-apx} for more details about \cref{prop:sublinear}.

\newcommand{\propSize}[0]{
Let $G = (V,E)$ be an $n$-vertex graph to which we have adjacency matrix query access. For any $\epsilon > 0$, there is a randomized algorithm that w.h.p. takes $\widetilde{O}(n/\epsilon^3)$ time and returns a number $\widetilde{g}$ where 
$
    \E_\pi[|\GMM{G, \pi}|] - \epsilon n \leq \widetilde{g} \leq \E_\pi[|\GMM{G, \pi}|].
$
}
\begin{proposition}[\cite{behnezhad2021}]\label{prop:sublinear}
\propSize{}
\end{proposition}


\begin{proof}[Proof of \cref{lem:implementation}]
The construction of $M'$, $V'$, $U$ in \cref{alg:static} takes $O(n)$ time since the edges of $M$ are given. But we cannot afford to explicitly construct $H$ as it may have $\Omega(n^2)$ edges. We get around this by using \cref{prop:sublinear} to estimate $g$ without constructing $H$.  To do this, it suffices to show that we can provide adjacency matrix access to $H$. Observe that an edge $e = (u, v)$ belongs to $H$ iff one of its endpoints is in $V'$, the other is in $U$, and additionally $(u, v) \in E$. This can be checked for any $(u, v)$ in $O(1)$ time since we explicitly store $V', U$ and have adjacency matrix access to $G$. Thus, we can apply \cref{prop:sublinear} on graph $H$ in $\widetilde{O}(n/\epsilon^3)$ time and obtain $g - \epsilon n \leq \widetilde{g} \leq g$. Using $\widetilde{g}$ instead of $g$ in the output of \cref{alg:static} proves the lemma.
\end{proof}

Next, we turn to prove the approximation guarantee of \cref{alg:static}. To do this, we start with the following useful proposition, proved by \citet*{KonradMM12} in the context of streaming algorithms. It shows that a randomized greedy maximal matching obtains a much better than 1/2-approximation if it is run on a vertex-sampled subgraph of a bipartite graph. 

\begin{proposition}[{\cite[Theorem~3]{KonradMM12}}]\label{prop:greedy-random-vertex}
    Let $0 < p \leq 1$, let $G=(A, B, E)$ be a bipartite graph, let $A' \subseteq A$ include each vertex of $A$ independently with probability $p$, and let $H$ be the induced subgraph of $G$ on vertex-set $A' \cup B$. Then for any permutation $\pi$ over the edge-set of $H$, $$\E_{A'}[|\GMM{H, \pi}|] \geq \frac{p}{1+p} \cdot \mu(G).$$
\end{proposition}

We are now ready to analyze the approximation ratio. The following lemma gives a lower bound on the expected value of the estimate $\widetilde{\mu}'$.

\begin{lemma}\label{lem:mu-lb}
    For \cref{alg:static}, it holds that $\E[\widetilde{\mu}'] \geq (2-\sqrt{2}) \mu(G)$.
\end{lemma}
\begin{myproof}
    Let $L$, $R$ be the two vertex parts in $G$. Define the following bipartite induced subgraphs of $G$:
    $$
        F_L := G[V(M) \cap L, U \cap R], \qquad F_R := G[V(M) \cap R, U \cap L],
    $$
    $$
        H_L := G[V(M') \cap L, U \cap R], \qquad H_R := G[V(M') \cap R, U \cap L].
    $$
    
    Let us fix an arbitrary maximum matching $M^\star$ of $G$. Since $M^\star$ is a maximum matching of $G$, we have exactly $|M^\star| - |M| = \mu(G) - |M|$ augmenting paths for $M$ in $M^\star \oplus M$. Let $L_1$ be the number of length one augmenting paths in $M^\star \oplus M$ for $M$, noting that $L_1 = 0$ here since $M$ is maximal.\footnote{The reason that we define $L_1$ is that later when we switch to adaptive adversaries, $M$ will not be maximal.} Note that every augmenting path of length at least three in $M^\star \oplus M$ has one of its endpoint edges in $F_L$ and the other in $F_R$. Moreover, these endpoint edges form a matching as they all belong to $M^\star$. Hence,
    \begin{equation}\label{eq:mtlr-1290}
        \mu(F_L) \geq \mu(G) - |M| - L_1, \qquad\qquad \mu(F_R) \geq \mu(G) - |M| - L_1.
    \end{equation}
    
    Moreover, note that since $M'$ is a random subsample of $M$, then each vertex in $V(M) \cap L$ belongs to $V(M') \cap L$ independently from the vertices in $L$ (but not $R$) with probability $p$. As a result, $H_L$ is an induced subgraph of $F_L$ which includes all the vertices in one part, and $p$ fraction of the vertices in the other part independently. Applying \cref{prop:greedy-random-vertex}, we thus get that
    \begin{equation*}
        \E_{M'}[|\GMM{H_L, \pi}|] \geq \frac{p}{1+p} \cdot \mu(F_L) \stackrel{(\ref{eq:mtlr-1290})}{\geq} \frac{p}{1+p} \cdot (\mu(G) - |M| - L_1).
    \end{equation*}
    With essentially the same proof, we get the same lower bound for $\E_{M'}[|\GMM{H_R, \pi}|]$.
    Since $G$ is bipartite, $H_L$ and $H_R$ are vertex disjoint. Combined with $H = H_L \cup H_R$ this implies that
    $$
        \E_{M'}[|\GMM{H, \pi}|] = \E_{M'}[|\GMM{H_L, \pi}|] + \E_{M'}[|\GMM{H_R, \pi}|] \stackrel{}{\geq} \frac{2p}{1+p} (\mu(G) - |M| - L_1).
    $$
    This in turn implies that 
    \begin{equation}\label{eq:mtlrc-31999}
        \E_{M'}[g] = \E_{M',\pi}[|\GMM{H, \pi}|] \geq \frac{2p}{1+p}(\mu(G)-|M|-L_1).
    \end{equation} 
    Taking expectation over $M'$ in the output $\widetilde{\mu}'$ of \cref{alg:static}, we get
    \begin{flalign}
        \nonumber \E_{M'}[\widetilde{\mu}'] &= \E_{M'}[|M| + \max\{0, g - |M'|\}] = |M| + \max\{0, \E_{M'}[g] - p|M|\} \\
        \nonumber &\geq (1-p)|M| + \frac{2p}{1+p} (\mu(G)-|M|-L_1) \tag{By (\ref{eq:mtlrc-31999}).}\\
        &= \left(1 - p - \frac{2p}{1+p} \right)|M| + \frac{2p}{1+p} \mu(G) - \frac{2p}{1+p} L_1.\label{eq:hgcsii-13289}
    \end{flalign}
    Since $p = \sqrt{2} - 1$, we get $1-p - \frac{2p}{1+p} = 0$ and $\frac{2p}{1+p} = 2 - \sqrt{2}$. Combined with $L_1 = 0$, this implies 
    $$
        \E_{M'}[\widetilde{\mu}'] \geq (2-\sqrt{2}) \mu(G).\qedhere
    $$
\end{myproof}

Next, we show that $\widetilde{\mu}'$ does not over-estimate $\mu(G)$.

\begin{lemma}\label{lem:mu-ub}
    For \cref{alg:static}, it holds with probability 1 that $\widetilde{\mu}' \leq \mu(G)$.
\end{lemma}
\begin{myproof}
    Take an arbitrary matching $S$ in graph $H$ of \cref{alg:static}. Let $M''$ be the subset of edges in $M'$ whose both endpoints are matched by $S$.  By definition of $H$, every edge in $S$ has one endpoint in $V' = V(M')$ and one endpoint in $U = V \setminus V(M)$. This means that $M \oplus S$ includes at least $|M''|$ length three augmenting paths for $M$, and so
    \begin{equation}\label{eq:mlch-3123000}
    \mu(G) \geq |M| + |M''|.
    \end{equation}
    The fact that any edge in $S$ has exactly one vertex in $V(M')$ implies that the number of vertices in $V(M')$ unmatched by $S$ is at most $|V(M')| - |S| = 2|M'| - |S|$. As such, there are at least $|M'| - (2|M'| - |S|) = |S| - |M'|$ edges in $M'$ whose both endpoints are matched by $S$. Hence,  $|M''| \geq |S| - |M'|$. Plugging this into \cref{eq:mlch-3123000}, we get that 
    \begin{equation}\label{eq:mlc-nrl00288}
        \mu(G) \geq |M| + |S| - |M'|.
    \end{equation}
    Now instead of an arbitrary matching, take $S$ to be a maximum matching of $H$. Note, in particular, that $|S| \geq \E_\pi[|\GMM{H, \pi}| \mid M'] = g$. Plugging this into \cref{eq:mlc-nrl00288}, and noting also that $\mu(G) \geq |M|$ since $M$ is a matching in $G$, we get that
    \begin{equation*}
        \mu(G) \geq \max\{|M|, |M| + g - |M'|\} = |M| + \max\{0, g - |M'|\} = \widetilde{\mu}'.\qedhere
    \end{equation*}
\end{myproof}

We are now ready to complete the proof of \cref{lem:A-for-bipartite}, which as discussed also proves \cref{thm:main} for bipartite graphs.

\begin{proof}[Proof of \cref{lem:A-for-bipartite}]
    The data structures that we store, as discussed, take only $(\poly\log n)$ worst-case time to maintain. When the algorithm is queried, we return the output $\widetilde{\mu}''$ of \cref{lem:implementation}. It takes $\widetilde{O}(n/\epsilon^3)$ time to produce this by \cref{lem:implementation}, which is the desired query time of \cref{lem:A-for-bipartite}. Moreover, for the approximation ratio, we have
    $$
        (2-\sqrt{2}) \mu(G) - \epsilon n
        \stackrel{\text{\cref{lem:mu-lb}}}{\leq}
        \E[\widetilde{\mu}'] - \epsilon n
        \stackrel{\text{\cref{lem:implementation}}}{\leq} 
        \E[\widetilde{\mu}''] \stackrel{\text{\cref{lem:implementation}}}{\leq} 
        \E[\widetilde{\mu}']
        \stackrel{\text{\cref{lem:mu-ub}}}{\leq}
        \mu(G).
    $$
   This completes the proof of \cref{lem:A-for-bipartite}.
\end{proof}

\paragraph{Adaptive Adversaries:} The only part of the algorithm discussed above that requires the oblivious adversary assumption is the maintanence of maximal matching $M$. While it is not known whether a maximal matching can be maintained in $(\poly\log n)$ time against adaptive adversaries, there are algorithms for maintaining a (non-maximal) $(1/2-\epsilon)$-approximate matching against adaptive adversaries in $\poly\log n$ time for any fixed $\epsilon > 0$ \cite{Wajc-STOC20}. Let $M$ be one such matching maintained. Observe that in our analysis, we assumed that $L_1 = 0$ which no longer holds for non-maximal $M$. But note that $L_1$ being large is not actually that bad as we can simply run the sublinear algorithm of \cite{behnezhad2021} on $G[V \setminus V(M)]$ to get a half-approximation of $L_1$, obtaining a matching of size at least $|M| + |L_1|/2$. Returning the larger of this estimate and the output of \cref{alg:static}, from \cref{eq:hgcsii-13289}, we get an estimator of size at least
$$
(1-\epsilon)\max\left\{ \frac{1}{2}\mu(G) + \frac{1}{2}|L_1|, \left(1 - p - \frac{2p}{1+p} \right)|M| + \frac{2p}{1+p} \mu(G) - \frac{2p}{1+p} L_1 \right\}.
$$
Setting $p=.3$, this is minimized for $L_1 = .084\mu(G)$, and the end result is a $.542$-approximation.




\begin{remark}
We note that the algorithm we employed in this section uses the graph's bipartiteness in a crucial way. In particular, our definitions of graphs $F_L$ and $F_R$, and the fact that they are vertex disjoint, crucially depends on the graph being bipartite. We later show in \cref{sec:general} how one can also beat 1/2-approximation for general graphs, albeit with a smaller improvement. Before that, we continue to focus on bipartite graphs, and prove \cref{thm:2/3} for them.
\end{remark}
\section{Beating Two-Thirds for Bipartite Graphs}\label{sec:twothird-bipartite}

In this section, we prove \cref{thm:2/3}. Our plan, similar to \cref{sec:bipartite}, is to turn a semi-dynamic algorithm into a fully-dynamic one. However, instead of \cref{lem:A-to-B}, which blows up the update-time by a $\poly\log n$ factor, we use a more refined variant of it that doesn't lose the $\poly\log n$ factor. We can do this in this section because our semi-dynamic algorithm will actually guarantee a multiplicative approximation instead of a multiplicative-additive one.

Let us now state the guarantee of our semi-dynamic algorithm in this section.

\begin{lemma}\label{lem:A-for-twothird-bipartite}
    For any $n$-vertex fully dynamic bipartite graph $G$ of maximum degree at most $\Delta$, there is a (randomized) data structure $\mc{A}$ that takes $O(\sqrt{\Delta})$ worst-case update-time, and upon being queried takes $O(n\sqrt{\Delta}) + \widetilde{O}(n)$ time to produce a number $\widetilde{\mu}$ such that for some absolute constant $\delta_0 > 1.8 \times 10^{-6}$, it holds w.h.p. that
    $
        (\frac{2}{3} + \delta_0) \cdot \mu(G) \leq \widetilde{\mu} \leq \mu(G).
    $ 
\end{lemma}

\cref{lem:A-for-twothird-bipartite} suffices to prove \cref{thm:2/3}:

\begin{proof}[Proof of \cref{thm:2/3}]
    Let $\epsilon$ be small enough that $\delta_0 - \epsilon > 10^{-6}$, where $\delta_0 > 1.8 \times 10^{-6}$ is as in \cref{lem:A-for-twothird-bipartite}. We run the semi-dynamic data structure $\mc{A}$ of \cref{lem:A-for-twothird-bipartite} in the background, which takes $O(\sqrt{\Delta})$ worst-case update-time against adaptive adversaries. Suppose that we call the oracle of algorithm $\mc{A}$ once to produce the estimate $\widetilde{\mu}$. Observe that within the next $\epsilon \widetilde{\mu}$ updates, the maximum matching size of $G$ changes by at most $\epsilon \widetilde{\mu}$, even against adaptive adversaries. Hence, we can use this oracle in a lazy way: we call it once, return $(1-\epsilon)\widetilde{\mu}$ as the output, do not change the output for $\epsilon \widetilde{\mu}$ updates, then repeat. This way, the \underline{amortized} time spent on the oracle is indeed $\Big(O(n \sqrt{\Delta}) + \widetilde{O}(n)\Big) / \epsilon n = O(\sqrt{\Delta})+\poly\log n$.  This can be turned into a worst-case bound using a well-known `spreading' technique (see \cite{GuptaPeng-FOCS13}). Letting $T$ denote the time spent by the oracle, the idea is to spread its computation over $\Theta(\epsilon \widetilde{\mu})$ updates, each performing $\Theta(T/\epsilon \widetilde{\mu})$ operations of the algorithm. When the computation of the oracle finishes, we change our output and immediately start spreading its next execution. Finally, we note that by losing only a $(1-\epsilon)$ factor in the approximation, one can assume that $\Delta \leq O(\sqrt{m}/\epsilon)$ using a marking algorithm of \cite{Solomon-ITCS18} (see \cite[Section~5]{GrandoniSSU-SOSA22} or \cite[Section~4.7]{BehnezhadK-SODA22} for how this can be used for dynamic algorithms). Thus, overall, the update-time that we get is  $O(\min\{\sqrt{\Delta}, m^{1/4}\} + \poly\log n) = O(\sqrt{n})$ in the worst-case.
    
    The approximation guarantee follows immediately from the high probability and multiplicative guarantee of the semi-dynamic algorithm in \cref{lem:A-for-twothird-bipartite}.
\end{proof}

In order to prove \cref{lem:A-for-twothird-bipartite}, we build on the edge-degree constrained subgraph (EDCS) of \citet*{BernsteinS-ICALP15}. An EDCS can be used to maintain the edges of a $(2/3 - \epsilon)$-approximate maximum matching in $O(\sqrt{n} \poly(1/\epsilon))$ time \cite{BernsteinSteinSODA16}, and this bound is known to be tight for it. We show how to go beyond 2/3-approximation by proving a certain characterization of the tight instances of EDCS, where it only obtains a 2/3-approximation. We first give some background on EDCS in \cref{sec:EDCS-background}, then state our characterization for its tight instances in \cref{sec:EDC-char-statement}, and then use this characterization to prove \cref{lem:A-for-twothird-bipartite}.

\subsection{Background on Edge-Degree Constrained Subgraphs (EDCS)} \label{sec:EDCS-background}
The EDCS is a matching sparsifier introduced by \citet*{BernsteinS-ICALP15}, defined as follows:

\begin{definition}[\cite{BernsteinS-ICALP15}]
For any $\beta > \beta_- \geq 1$, a subgraph $H$ of $G$ is a {\em $(\beta, \beta_-)$-EDCS} of $G$ if
\begin{itemize}[topsep=0pt,itemsep=0pt]
    \item for all edges $(u, v) \in H$, \hspace{0.77cm} $\deg_H(u) + \deg_H(v) \leq \beta$, and
    \item for all edges $(u, v) \in G \setminus H$, \hspace{0.1cm} $\deg_H(u) + \deg_H(v) \geq \beta_-$.
\end{itemize}
\end{definition}

It is known that for any integers $\beta > \beta_- \geq 1$, any graph $G$ has a $(\beta, \beta_-)$-EDCS. The main property of EDCS, first proved in \cite{BernsteinS-ICALP15,BernsteinSteinSODA16} and further refined in \cite{AssadiBernstein19,BehnezhadEDCS}, is that if $\beta \geq 1/\epsilon$ and $\beta_- \geq (1-\epsilon)\beta$, then $\mu(G) \geq (\frac{2}{3} - O(\epsilon))\mu(G)$. Moreover, this guarantee is  tight. That is, the $2/3$ factor cannot be replaced by a larger constant even if $\beta_- = \beta - 1 = \Theta(n)$ \cite{BernsteinS-ICALP15}. We refer interested readers to the paper of \citet*{AssadiBernstein19} for an excellent overview of EDCS and its applications across various models.

Several algorithms are known for maintaining an EDCS in dynamic graphs \cite{BernsteinS-ICALP15,BernsteinSteinSODA16,BehnezhadK-SODA22,GrandoniSSU-SOSA22,Kiss-ITCS22}. Here we state a simple and clean algorithm of \citet*{GrandoniSSU-SOSA22} which is deterministic, works for general graphs, and its update-time bound holds in the worst-case.

\begin{proposition}[\cite{GrandoniSSU-SOSA22}]\label{prop:maintain-EDCS}
    Let $G$ be a fully dynamic graph and let $\Delta$ be a fixed upper bound on its maximum degree. For some $\beta = \Theta(\sqrt{\Delta}\poly(1/\epsilon))$, one can maintain the edges of a $(\beta, (1-\epsilon)\beta)$-EDCS $H$ of a graph $G$, as well as the edges of a $(1-\epsilon)$-approximate maximum matching $M_H$ of $H$ deterministically in worst-case update-time $O(\sqrt{\Delta} \poly(1/\epsilon))$.
\end{proposition}

Combined with the guarantee above on the approximation ratio of EDCS, we get a fully dynamic algorithm that maintains a $(\frac{2}{3}-\epsilon)$-approximate matching of $G$ in $O(\sqrt{\Delta}\poly(1/\epsilon))$ update-time. 

\subsection{A Structural Result on Tight Instances of EDCS}\label{sec:EDC-char-statement}

Recall that our goal in \cref{lem:A-for-twothird-bipartite} is to go beyond 2/3-approximation. Towards this, we prove a structural result on the tight instances of EDCS, and then use it to break 2/3-approximation. \cref{lem:EDCS-char} is this characterization. In essence, it shows that when an EDCS $H$ with the right parameters does not include a strictly larger than $2/3$-approximation, then there is an almost 2/3-approximate matching $M$ in $H$ that is far from being maximal (not maximum) for $G$. More precisely, \cref{lem:EDCS-char} guarantees that there must be a matching of size nearly $\mu(G)/3$ in $G$ among vertices left unmatched by $M$.
    
    \newcommand{\vmid}[0]{\ensuremath{V_{\textsf{mid}}}}
    \newcommand{\vlow}[0]{\ensuremath{V_{\textsf{low}}}}
    
    \begin{lemma}[\textbf{On Tight Instances of EDCS in Bipartite Graphs}]\label{lem:EDCS-char}
        Let $\epsilon \in (0, \frac{1}{120})$, let $H$ be a $(\beta, (1-\epsilon)\beta)$-EDCS of a bipartite graph $G$ for any $\beta > (1-\epsilon)\beta \geq 1$. Let $\vmid$ and $\vlow$ respectively include vertices $v$ such that $\deg_H(v) \in [.4\beta, .6\beta]$ and $\deg_H(v) \in [0, .2\beta]$. Let $H' := H[\vlow, \vmid]$ be the subgraph of $H$ on edges with one endpoint in $\vmid$ and one in $\vlow$. Let $\delta \in (2\epsilon, \frac{1}{60})$ be any parameter. If $\mu(H) \leq (\frac{2}{3} + \delta) \mu(G)$ then the following hold:
        \begin{enumerate}[itemsep=0pt,topsep=5pt,label=$(P\arabic*)$,leftmargin=40pt]
            \item $\mu(H') \geq \big(\frac{2}{3} - 120\sqrt{\delta}\big)\mu(G)$.
            \item For any matching $M$ in $H'$, $\mu(G[\vmid \setminus V(M)]) \geq (\frac{1}{3}-800 \cdot\delta)\mu(G)$.
            \item $|\vmid| < 8\mu(G)$.
        \end{enumerate}
    \end{lemma}
    
    The proof of \cref{lem:EDCS-char} is quite involved, so we defer it to \cref{sec:EDCS-bipartite}. It is worth noting that we did not attempt to optimize the constants in \cref{lem:EDCS-char}.
    
    \paragraph{Comparison to a Characterization of \citet*{AssadiBehnezhad21}:} Prior to this work, another characterization of tight instances of EDCS was given in \cite{AssadiBehnezhad21} in the context of random order streaming algorithms. The characterization of \cite{AssadiBehnezhad21} implies existence of a nearly $\frac{2}{3}$-approximate matching in subgraph $H[\vmid]$ when $\mu(H)$ is not strictly larger than $\frac{2}{3}\mu(G)$.  Roughly speaking, \cite{AssadiBehnezhad21} showed that this matching in $H[\vmid]$ can be found early on in the stream, and showed how the rest of the stream could be used to discover many length-five augmenting paths for it and beat $2/3$-approximation for random-order streams. We do not know how to find or estimate the number of length-five augmenting paths efficiently in the dynamic setting. Fortunately, the guarantee of \cref{lem:EDCS-char} helps us avoid them all togeher, and only focus on length-one augmenting paths instead. The following figure illustrates how the two guarantees differ. It is a tight instance of EDCS, where the vertices in $\vmid$ all have degree $\beta/2$ in $H$. The dashed edges are missed from the EDCS while all other edges are present. It can be confirmed that these missed edges alone imply that $\mu(H) \leq \frac{2}{3} \mu(G)$. The blue matching $M_L$ on the left is the matching in $H[\vmid]$ used by \cite{AssadiBehnezhad21}. The green matching $M_R$ on the right is the matching of $H'$ that \cref{lem:EDCS-char} guarantees to exist. While both are of size $\frac{2}{3}\mu(G)$, the key difference is that each edge of $M_R$ has exactly one endpoint in $\vmid$, whereas both endpoints of all edges of $M_L$ are in $\vmid$. Consequently, while $M_L$ is nearly maximal for $G$ and only leaves length-five augmenting paths, $M_R$ is far from being maximal for $G$ and all the dashed edges can be directly added to it. 
    
    \begin{figure}[h]
        {\centering
        \includegraphics[scale=0.45]{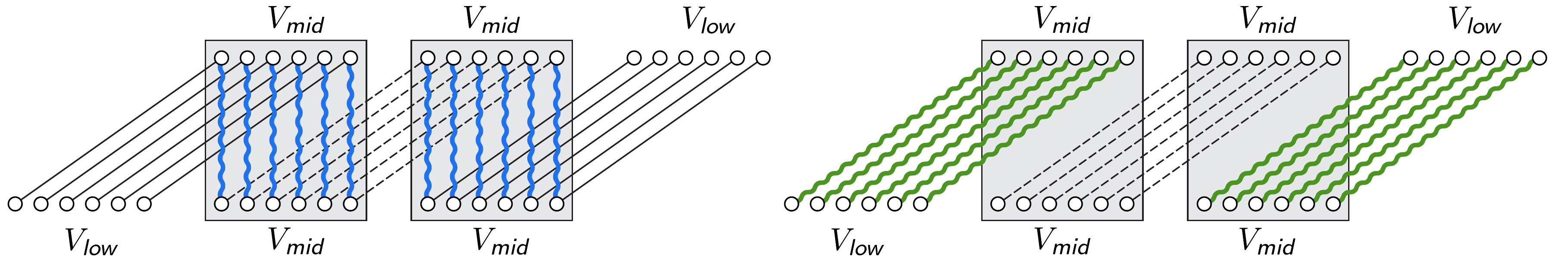}
        }
        {\color{white}{.}}
        \caption{\small Comparison of the matching implied in the work of \citet*{AssadiBehnezhad21} (the blue matching on the left), and the matching guaranteed by \cref{lem:EDCS-char} (the green matching on the right).}
    \end{figure}
    
    \subsection{The Semi-Dynamic Algorithm (Proof of \cref{lem:A-for-twothird-bipartite}) via \cref{lem:EDCS-char}}
    
    In this section, we show how the characterization of \cref{lem:EDCS-char} can be used to prove our semi-dynamic algorithm in \cref{lem:A-for-twothird-bipartite} for beating 2/3-approximation in bipartite graphs.
    
    \smparagraph{Data Structures:} As in \cref{sec:bipartite}, we start by describing the data structures maintained by the semi-dynamic algorithm. We maintain the adjacency matrix of the graph $G$. Additionally, we run \cref{prop:maintain-EDCS} to maintain the edges of a $(\beta, (1-\epsilon)\beta)$-EDCS $H$ of $G$ for $\beta = \Theta(\sqrt{\Delta}\poly(1/\epsilon))$, as well as a $(1-\epsilon)$-approximate maximum matching $M_H$ of $H$ in $O(\sqrt{\Delta}\poly(1/\epsilon))$ worst-case update-time, where we will set $\epsilon$ to be a sufficiently small absolute constant.
    
    It remains to show how the query algorithm works. That is, how we produce the $(\frac{2}{3}+\delta_0)$-approximate estimate $\widetilde{\mu}$ for  $\mu(G)$.
    
    \smparagraph{The Query Algorithm:} When queried, we run (an estimate) of the following algorithm:
    
    \begin{alglist}{}{alg:query-twothirds-bipartite}
    \begin{enumerate}[leftmargin=15pt, itemsep=0pt]
        \item Let $H$ be the $(\beta, (1-\epsilon)\beta)$-EDCS that we maintain in our data structure.
        
        \item Let $\vmid := \{v \mid \deg_H(v) \in [.4\beta, .6\beta]\}$ and let $\vlow := \{ v \mid \deg_H(v) \in [0, .2\beta]\}$.
        
        \item Let $H' := H[\vlow, \vmid]$.
        
        \item Find a $(1-\epsilon)$-approximate maximum matching $M_{H'}$ in $H'$.
        
        \item Let $F := G[\vmid \setminus V(M_{H'})]$. (We will not construct $F$ explicitly.)
        
        \item Let $g := |\GMM{F, \pi}|$ for an arbitrary permutation $\pi$.
        
        \item Return $\widetilde{\mu}' := \max\{|M_H|, |M_{H'}| + g\}$, where $M_H$ is the $(1-\epsilon)$-approximate matching of $H$ that we maintain in our data structures.
    \end{enumerate}
    \end{alglist}
    
    Let us show that the output of \cref{alg:query-twothirds-bipartite} can be estimated efficiently, in the desired time of \cref{lem:A-for-twothird-bipartite}.
    
    \begin{lemma}\label{lem:implementation-23-bp}
        Let $\widetilde{\mu}'$ be as in \cref{alg:query-twothirds-bipartite}. There is an algorithm that takes $O(n \sqrt{\Delta} \poly(1/\epsilon)) + \widetilde{O}(n \poly(1/\epsilon))$ time and returns a number $\widetilde{\mu}''$ such that w.h.p. $(i)$ $\widetilde{\mu}'' \leq \widetilde{\mu}'$, $(ii)$ $\widetilde{\mu}'' \geq |M_H|$, and $(iii)$ $\widetilde{\mu}'' \geq |M_{H'}| + g - \epsilon |\vmid|$.
    \end{lemma}
    \begin{myproof}
        First, note that $H$, $\vmid$, $\vlow$, and $H'$ can all be explicitly constructed in time linear in the size of $H$, which is $O(n \beta)$ (since any $(\beta, \cdot)$-EDCS has maximum degree at most $\beta$). Moreover, a $(1-\epsilon)$-approximate matching of an $m$-edge graph can be found in $O(m/\epsilon)$ time using the algorithm of Hopcroft and Karp \cite{HopcroftK73}. Therefore, since $H'$ is a subgraph of $H$ and thus also has at most $O(n\beta)$ edges, it takes $O(n \beta/\epsilon)$ time to construct matching $M_{H'}$ of $H'$. Finally, instead of constructing $F$ explicitly and computing $g$ for it, we use \cref{prop:sublinear}. Note that since we have adjacency matrix access to graph $G$, we can provide adjacency matrix to its subgraph $F$ as well. Moreover, since $F$ has only $|\vmid|$ vertices by its definition,  \cref{prop:sublinear} takes $O(|\vmid|/\epsilon^3) = \widetilde{O}(n/\epsilon^3)$ time to produce $\widetilde{g}$ such that w.h.p. $\E_\pi|\GMM{F, \pi}| - \epsilon |\vmid| \leq \widetilde{g} \leq \E_\pi|\GMM{F, \pi}|$. Using this instead of $g$ in \cref{alg:query-twothirds-bipartite}, we obtain the estimate $\widetilde{\mu}''$ that satisfies guarantees $(i), (ii), (iii)$ of the lemma. The final running time of the algorithm is $O(n\beta/\epsilon) + \widetilde{O}(n/\epsilon^3)$ which is the desired bound of the lemma given that $\beta = \Theta(\sqrt{\Delta} \poly(1/\epsilon))$.
    \end{myproof}
    
    Next, we focus on the approximation ratio of the output of \cref{alg:query-twothirds-bipartite}. The following lemma essentially lower bounds the output of \cref{lem:implementation-23-bp} by $(\frac{2}{3}+\delta_0)\mu(G)$.
    
    \begin{lemma}\label{lem:23-mu-lb-bp}
        Let $\delta_0 = 1.8 \times 10^{-6}$. For \cref{alg:query-twothirds-bipartite}, at least one of the inequalities $|M_H| \geq (\frac{2}{3} + \delta_0) \mu(G)$ and $|M_{H'}| + g - \epsilon |\vmid| \geq (\frac{2}{3} + \delta_0) \mu(G)$ must hold.
    \end{lemma}
    \begin{myproof}
        Let $\delta = 1.9 \times 10^{-6}$ and suppose that we set $\epsilon < \delta/100$. If $|M_H| \geq (1-\epsilon)(\frac{2}{3} + \delta)\mu(G)$, then 
        $$
            |M_H| \geq \Big(\frac{2}{3} + \delta - \epsilon\Big) \mu(G) \geq \Big(\frac{2}{3} + 0.99\delta\Big) \mu(G) \geq \Big(\frac{2}{3} + \delta_0\Big) \mu(G),
        $$
        which is exactly the first inequality. So let us assume that $|M_H| < (1-\epsilon)(\frac{2}{3} + \delta) \mu(G)$. Given that $|M_H| \geq (1-\epsilon) \mu(H)$ for being a $(1-\epsilon)$-approximate matching, we get that $\mu(H) < (\frac{2}{3}+\delta)\mu(G)$. Plugging this into the characterization of \cref{lem:EDCS-char} for the tight instances of EDCS (noting in particular that $\delta$ and $\epsilon$ satisfy the range constraints), we get:
        \begin{enumerate}[itemsep=0pt,topsep=5pt,label=$(P\arabic*)$,leftmargin=40pt]
            \item $\mu(H') \geq \big(\frac{2}{3} - 120\sqrt{\delta}\big)\mu(G)$. This, in particular, implies that $|M_{H'}| \geq (1-\epsilon)(\frac{2}{3} - 120\sqrt{\delta})\mu(G)$.
            \item For any matching $M$ in $H'$, $\mu(G[\vmid \setminus V(M)]) \geq (\frac{1}{3}-800 \cdot\delta)\mu(G)$. Using $M_{H'}$ for $M$ in this statement, this means that $\mu(F) \geq (\frac{1}{3} - 800\delta)\mu(G)$. Since $g$ is at least half the size of $\mu(F)$ for being the size of a maximal matching, we thus get that $g \geq \frac{1}{2}(\frac{1}{3} - 800 \delta)\mu(G)$.
            \item $|\vmid| \leq 8 \mu(G)$.
        \end{enumerate}
        From this, we can infer the second inequalit as follows:
        \begin{flalign*}
            |M_{H'}| + g - \epsilon |\vmid| &\geq (1-\epsilon)\Big(\frac{2}{3} - 120\sqrt{\delta}\Big)\mu(G) + \frac{1}{2}\Big(\frac{1}{3} - 800 \delta\Big)\mu(G) - 8\epsilon \mu(G) \\
            &\geq (1-\epsilon)\Big(\frac{2}{3} - 120\sqrt{\delta} + \frac{1}{6} - 400 \delta - 8\epsilon \Big)\mu(G)\\
            &\geq (1-\epsilon)\Big(\frac{2}{3} + 1.9 \times 10^{-6} - 8\epsilon \Big) \mu(G) \tag{Since $\delta = 1.9 \times 10^{-6}$.}\\
            &\geq \Big(\frac{2}{3} + \delta_0\Big)\mu(G). \tag{Since $\delta_0 = 1.8 \times 10^{-6}$ and $\epsilon < \delta/100$.}
        \end{flalign*}
        This completes the proof.
    \end{myproof}
    
    Next, we show that the output $\widetilde{\mu}'$ of \cref{alg:query-twothirds-bipartite} does not overestimate the matching size.
    
    \begin{observation}\label{obs:23-mu-ub-bp}
        For \cref{alg:query-twothirds-bipartite}, it holds with probability 1 that $\widetilde{\mu}' \leq \mu(G)$.
    \end{observation}
    \begin{myproof}
        \cref{alg:query-twothirds-bipartite} sets $\widetilde{\mu}' = \max\{|M_H|, |M_{H'}| + g\}$. Clearly $|M_H| \leq \mu(G)$ since $M_H$ is a matching of $H \subseteq G$. On the other hand, take the matching $M_F = \GMM{F, \pi}$ and note that $g$ is defined to be $|M_F|$ in \cref{alg:query-twothirds-bipartite}. Since $F = G[\vmid \setminus V(M_{H'})]$, the edges of $M_F$ are vertex disjoint from $M_{H'}$. Hence, $M_F \cup M_{H'}$ is a matching of $G$, and so $|M_F \cup M_{H'}| = |M_{H'}|+g \leq \mu(G)$.
    \end{myproof}
    
    We are now ready to finish the proof of \cref{lem:A-for-twothird-bipartite}.
    
    \begin{proof}[Proof of \cref{lem:A-for-twothird-bipartite}]
        The data structures that we store, as discussed, take $O(\sqrt{\Delta} \poly(1/\epsilon)) = O(\sqrt{\Delta})$ worst-case update time. When the oracle is called, we call the algorithm of \cref{lem:implementation-23-bp} and return its estimate $\widetilde{\mu}''$. Its running time is $O(n\sqrt{n})+\widetilde{O}(n)$ since we set $\epsilon$ to be an absolute constant. This is the desired running time in \cref{lem:A-for-twothird-bipartite}. For the approximation, first note by \cref{lem:implementation-23-bp} and \cref{obs:23-mu-ub-bp}, that we have $\widetilde{\mu}'' \leq \mu(G)$ w.h.p. Moreover, by \cref{lem:23-mu-lb-bp} and \cref{lem:implementation-23-bp}, we have $\widetilde{\mu}'' \geq (\frac{2}{3}+\delta_0)\mu(G)$ w.h.p. This completes the proof.
    \end{proof}

\subsection{Proof of \cref{lem:EDCS-char}; the Characterization for Tight EDCS Instances}\label{sec:EDCS-bipartite}

\newcommand{\missingMstar}[0]{\ensuremath{\overline{M}^\star}}
\newcommand{\Abar}[0]{\ensuremath{\overline{A}}}
\newcommand{\Bbar}[0]{\ensuremath{\overline{B}}}

In this section, we prove \cref{lem:EDCS-char}. Let $L$ and $R$ with $|L| = |R| = n$ be the two vertex parts for graph $G$. We use the following standard extension of the Hall's theorem.

\begin{proposition}[Extended Hall's Theorem \cite{Hall35}]\label{prop:halls}
	Let $G=(L,R,E)$ be a bipartite graph and $|L| = |R| = n$. Then,
	\[
		\max (|A| - |N(A)|) = n - \mu(G), 
	\]
	where $A$ ranges over $L$ or $R$, separately. We refer to such set $A$ as a {\em witness set}. 
\end{proposition}
 
Let $A$ be the Hall's witness of $H$ as defined in \cref{prop:halls}. Suppose w.l.o.g. that $A \subseteq L$. Define $\Abar := L \setminus A$, $B := N_H(A)$, $\Bbar := R \setminus B$. Fix a maximum matching $M^\star$ of graph $G$. Let $\missingMstar$ be the edges in $M^\star$ that have one endpoint in $A$ and one endpoint in $\Bbar$. Note that no edge of $\missingMstar$ can belong to $H$. We also define $S := V(\missingMstar)$, $W := (A \cup \Bbar) \setminus S$, and $T = \Abar \cup B$. See \cref{fig:hall}.
    
    \begin{figure}[h]
        \centering
        \includegraphics[scale=0.47]{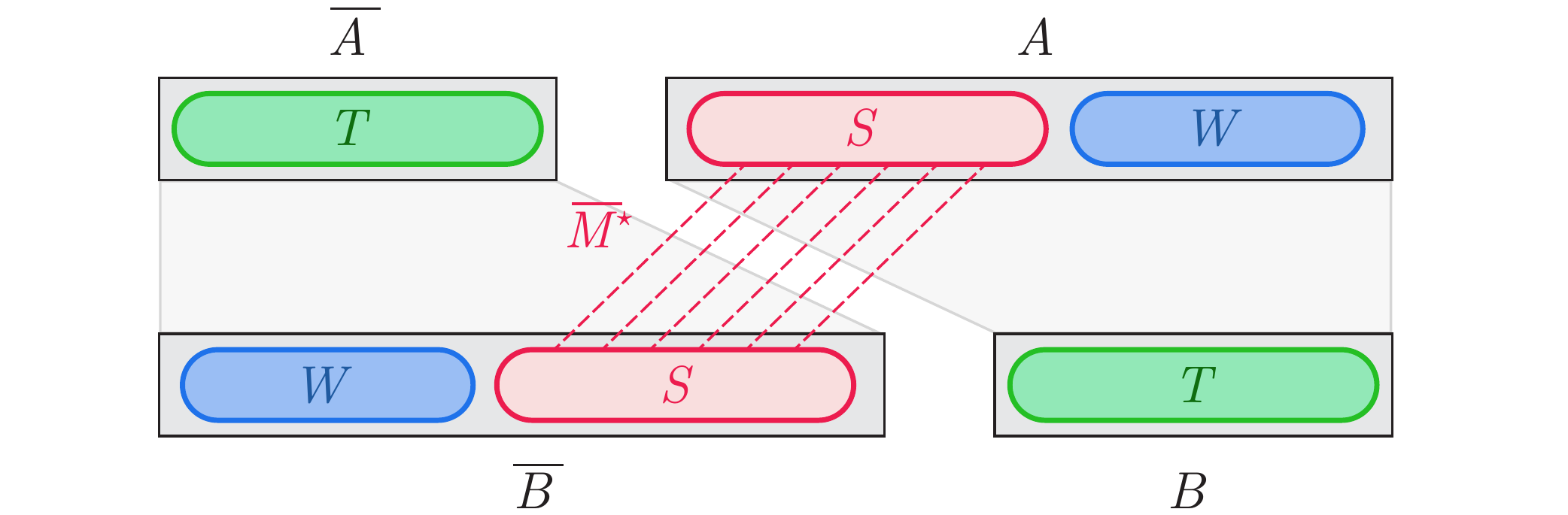}
        \caption{An illustration of the Hall's witness $A$ for $H$, along with the sets $B, \Abar, \Bbar, S, W$.}
        \label{fig:hall}
    \end{figure}
    
    \newcommand{\upT}[0]{\ensuremath{\hat{T}}}
    \newcommand{\upS}[0]{\ensuremath{\hat{S}}}
    \newcommand{\downS}[0]{\ensuremath{\check{S}}}
    \newcommand{\upW}[0]{\ensuremath{\hat{W}}}
    
    We show in this section that when $H$ does not include a larger than $2/3$-approximation, then (almost) all vertices in $S$ and $T$ must have degree very close to $\beta/2$ and so belong to $U$ (as defined in \cref{lem:EDCS-char}). Additionally, we show that very few vertices of $W$ belong to $U$.
    
    To proceed, we need some notation that we summarize below.
    \begin{itemize}
        \item $H_S, H_W$: We partition the edges of $H$ into two subgraphs $H_S$ and $H_W$. The edges between $S$ and $T$ belong to $H_S$ and the edges between $W$ and $T$ belong to $H_W$.
        \item $d(v), d_S(v), d_W(v)$: The degree of a vertex $v$ in graphs $H$, $H_S$, and $H_W$ respectively.
        \item $m$, $m_S$, $m_W$: The number of edges in $H$, $H_S$, and $H_W$ respectively.  (Other than the discussion of this section, we use $m$ to denote the number of edges of $G$ instead.)
        \item $\bar{d}_S(T) := m_S/|T|$, $\bar{d}_W(T) := m_W/|T|$, $\bar{d}(T) := m/|T|$: These are the average degrees of vertices in $T$ in graphs $H_S, H_W$, and $H$ respectively.
        \item $\bar{d}(S) := m_S / |S|$, $\bar{d}(W) := m_W / |W|$: The average degrees in $S$ and $W$ respectively.
        \item $\upT := \{ v \in T : d(v) > (1+\alpha) d_S(v)\}$ where $\alpha := 3\sqrt{\delta}$.
    \end{itemize}

    \paragraph{The characterization:} Our proofs proceed by assuming $\mu(H) < (\frac{2}{3} + \delta) \mu(G)$, and then proving some structural properties of the subgraph $H$. In particular, we show that if $\mu(H) < (\frac{2}{3} + \delta) \mu(G)$, then all the following must hold:
    \begin{itemize}[leftmargin=10pt]
        \item $|S \setminus \vmid| \leq 199 \cdot \delta \mu(G)$ (stated as \cref{cl:S/mid-small}). That is, almost all vertices in $S$ belong to $\vmid$.
        \item $|T \setminus \vmid| \leq 600 \cdot \delta \mu(G)$ (stated as \cref{cl:T/mid-small}). That is, almost all vertices in $T$ belong to $\vmid$.
        \item $|W \setminus \vlow| \leq  33 \cdot \sqrt{\delta} \mu(G)$ (stated as \cref{cl:WcapU-small}). That is, almost all vertices in $W$ belong to $\vlow$.
    \end{itemize}
    
    We note that the upper bound on $|S \setminus \vmid|$ can also be inferred from the characterization of \citet*{AssadiBehnezhad21}. The other two bounds require new ideas.
    
    We first show how these properties imply our desired \cref{lem:EDCS-char} for bipartite graphs.
    
    \newcommand{\notmissingMstar}[0]{\ensuremath{\widetilde{M}^\star}}
    
    \begin{proof}[Proof of \cref{lem:EDCS-char} for bipartite graphs]
        We prove $(P1)$, $(P2)$, and $(P3)$ of \cref{lem:EDCS-char} one by one.
        
        \paragraph{(P1): $\mu(H') \geq (\frac{2}{3} - 120\sqrt{\delta})\mu(G)$.}
        
        Take the maximum matching $M^\star$ of $G$ we fixed at the beginning of this section. Recall that $\missingMstar$ is the edges of $M^\star$ that go from $A$ to $\Bbar$, and also recall that we defined $S = V(\missingMstar)$. Now take the matching $\notmissingMstar := M^\star \setminus \missingMstar$. Since $M^\star$ is a matching and $V(\missingMstar) = S$, we get that no vertex in $S$ can be matched by $\notmissingMstar$. Hence, any edge in $\notmissingMstar$ must have one endpoint in $T$ and one endpoint in $W$. Call an edge $e \in \notmissingMstar$ {\em bad} if it has an endpoint in $W \setminus \vlow$ or an endpoint in $T \setminus \vmid$, and {\em good} otherwise. We have
        $$
            (\# \text{ of bad edges}) \leq |W \setminus \vlow| + |T \setminus \vmid| 
            \stackrel{\text{Claims \ref{cl:T/mid-small},\ref{cl:WcapU-small}}}{\leq} 33 \sqrt{\delta} \mu(G) + 600 \delta \mu(G) < 119\sqrt{\delta}\mu(G),
        $$
        where the last inequality holds since $\delta < 1/60$. From this, we get that
        \begin{flalign*}
            (\# \text{ of good edges}) &\geq |\notmissingMstar| - 119\sqrt{\delta}\mu(G)\\
            &= (\mu(G) - |S|/2) - 119\sqrt{\delta}\mu(G) \tag{Since $|M^\star| = \mu(G)$, $|S| = 2|\missingMstar|$, and $|\notmissingMstar| = |M^\star| - |\missingMstar|$.}\\
            &\geq \mu(G) - \frac{1}{2} (\frac{2}{3} + 3\delta)\mu(G) - 119\sqrt{\delta}\mu(G) \tag{$|S| \leq (\frac{2}{3} + 3\delta)\mu(G)$ by \cref{cl:S-2/3-muG}.}\\
            &> \left(\frac{2}{3} - 120 \sqrt{\delta}\right) \mu(G). \tag{Holds since $\delta < 1/60$.}
        \end{flalign*}
        
        Now observe that since a good edge $(u, v)$ has one endpoint in $\vlow$ and one endpoint in $\vmid$, we have $d(u) + d(v) \leq .2 \beta + .6 \beta = .8 \beta$. Since $H$ is a $(\beta, (1-\epsilon)\beta)$-EDCS of $G$ and $.8 \beta < (1-\epsilon)\beta$, excluding $(u, v)$ from $H$ would violate the second property of EDCS. Hence, all good edges in $\notmissingMstar$ must belong to $H$. Additionally, all good edges must also belong to subgraph $H'$ of $H$, since they have one endpoint in $\vmid$ and one in $\vlow$. Thus, we get that $\mu(H') \geq (\# \text{ of good edges}) \geq (\frac{2}{3} - 120\sqrt{\delta})\mu(G)$.
        
        \paragraph{(P2):} For any matching $M$ in $H'$, $\mu(G[\vmid \setminus V(M)]) \geq (\frac{1}{3}-800\delta)\mu(G)$.
        
        We say an edge $e \in \missingMstar$ is {\em wasted} if at least one of its endpoints is matched by $M$. Take a wasted edge $(u, v) \in \missingMstar$ and suppose that its endpoint $v$ is matched to some vertex $w$ in $M$. First, note that $v \in S$ since $S = V(\missingMstar)$ and note that $w \in T$ since $(u, w) \in M \subseteq H' \subseteq H$ and all edges of $S$ go to $T$ in $H$. Second, note that since $M$ is a matching in $H'$, one vertex of $(w, v)$ must belong to $\vmid$ and one to $\vlow$. From this, we get that for any wasted edge in $\missingMstar$, there is at least one dedicated vertex in $S \cup T$ that belongs to $\vlow$. Hence,
        \begin{flalign*}
            (\# \text{of wasted edges}) &\leq |S \cap \vlow| + |T \cap \vlow|\\
            &\leq |S \setminus \vmid| + |T \setminus \vmid|. \tag{Since $\vlow \cap \vmid = \emptyset$ by their definition.}\\
            &\leq 199 \delta \mu(G) + 600 \delta \mu(G) \tag{By \cref{cl:S/mid-small,cl:T/mid-small}.}\\
            &= 799\delta \mu(G).
        \end{flalign*}
        Since any edge of \missingMstar{} that is not wasted belongs to $G[\vmid \setminus V(M)]$, we get that
        \begin{flalign*}
            \mu(G[\vmid \setminus V(M)]) &\geq |\missingMstar| - 799 \delta \mu(G) = \frac{1}{2} |S| - 799 \delta \mu(G)\\
            &\geq \frac{1}{2} \left(\frac{2}{3} - 2\delta\right)\mu(G) - 799 \delta \mu(G) \tag{By \cref{cl:S-2/3-muG}.}\\
            &= \left(\frac{1}{3} - 800 \delta\right) \mu(G).
        \end{flalign*}
        
        \paragraph{(P3): {\normalfont $|\vmid| < 2\mu(G)$}.} We have
        \begin{flalign*}
            |\vmid| &= |\vmid \cap S| + |\vmid \cap T| + |\vmid \cap W| \tag{Since $S$, $T$, $W$ partition $V$.}\\
            &\leq |S| + |T| + |\vmid \cap W|\\
            &\leq |S| + |T| + |W \setminus \vlow| \tag{Since $\vlow \cap \vmid = \emptyset$.}\\
            &\leq |S| + |T| + 33\sqrt{\delta} \mu(G) \tag{By \cref{cl:WcapU-small}.}\\
            &< 8 \mu(G),
        \end{flalign*}
        where the last inequality follows from $|S| = 2|\missingMstar| \leq 2\mu(G)$, $|T| = \mu(H) \leq \mu(G)$, and $33\sqrt{\delta} < 4.3$ (since $\delta < 1/60)$.
    \end{proof} 
    
    Thus, it just remains to prove the three upper bounds above on $|S \setminus \vmid|$, $|T \setminus \vmid|$, and $|W \setminus \vlow|$. We continue with some basic guarantees of the Hall's witness in \cref{sec:xgg-1}, prove a parametrized guarantee on $\bar{d}_S(T)$ in \cref{sec:xgg-2}, prove some useful auxiliary claims in \cref{sec:xgg-3}, and then the upper bounds on $|S \setminus \vmid|$, $|T \setminus \vmid|$, and $|W \setminus \vlow|$ in \cref{sec:xgg-4}.
     
    \subsubsection{Basic Guarantees of Hall's Witness}\label{sec:xgg-1}
    
    \cref{cl:muH=T,cl:S-lb,cl:ratio-to-muG} below are all by now standard in analyzing EDCS. We provide the full proofs nonetheless to keep our discussion of this section self-contained.
    
    \begin{claim}\label{cl:muH=T}
        $\mu(H) = |T|$.
    \end{claim}
    \begin{myproof}
        We have $|T| = |\Abar| + |B| = n - (|A| - |B|) = n - (n - \mu(H)) = \mu(H)$, where the third equation follows from \cref{prop:halls} and the fact that $A$ is a witness set of $H$.
    \end{myproof}
    
    \begin{claim}\label{cl:S-lb}
        $|S| \geq 2(\mu(G) - \mu(H))$.
    \end{claim}
    \begin{myproof}
        Let $H^\star := M^\star \cup H$. Note that $\mu(H^\star) \geq |M^\star| = \mu(G)$. From \cref{prop:halls}, this means that $|A| - |N_{H^\star}(A)| \leq n - \mu(H^\star) \leq n - \mu(G)$. On the other hand, since $A$ is a witness set for $H$, we know that $|A| - |B| = n - \mu(H)$. Putting the two together, we get that $|N_{H^\star}(A)| - |B| \geq \mu(G) - \mu(H)$. From the construction of $H^\star$, this implies that at least $\mu(G) - \mu(H)$ edges of $A$ in $M^\star$ should go to $\Bbar$, implying that $|\missingMstar| \geq \mu(G) - \mu(H)$. The claim follows since $|S| = 2|\missingMstar|$ by definition.
    \end{myproof}
    
    \begin{claim}\label{cl:ratio-to-muG}
        $\mu(H) \geq \frac{2\bar{d}(S)}{2\bar{d}(S) + \bar{d}_S(T)} \cdot \mu(G)$.
    \end{claim}
    \begin{myproof}
        Observe from definition that $\bar{d}_S(T) \cdot |T| = m_S = \bar{d}(S) \cdot |S|$. Therefore
        $
            |T| = \frac{|S| \bar{d}(S)}{\bar{d}_S(T)}.
        $
        Plugging $|S| \geq 2(\mu(G) - \mu(H))$ of \cref{cl:S-lb}, and $\mu(H) = |T|$ of \cref{cl:muH=T}, we get that
        $$
            \mu(H) \geq \frac{2(\mu(G) - \mu(H)) \bar{d}(S)}{\bar{d}_S(T)} \qquad \Leftrightarrow \qquad \mu(H) \geq \frac{2\bar{d}(S)}{2\bar{d}(S) + \bar{d}_S(T)} \cdot \mu(G).\qedhere
        $$
        
    \end{myproof}
    
    \subsubsection{A Parametrized Lower Bound on the degrees of $T$ to $S$}\label{sec:xgg-2}
    
    The following \cref{lem:dT-ub} is our most technical lemma of this section. It gives a useful lower bound on the average degree $\bar{d}_S(T)$ based on a parameter $\gamma \geq 0$. We will later show that if $\mu(H)$ is not much larger than $\frac{2}{3}\mu(G)$, then $\gamma$ should be very close to zero, implying several useful properties on the structure of such tight instances.
    
    \begin{lemma}\label{lem:dT-ub}
        It holds that $\bar{d}_S(T) \leq (1 - \gamma)\beta - \bar{d}(S)$ where
        $$
             \gamma := \frac{\sum_{v \in S} (d(v) - \bar{d}(S))^2 + \frac{1}{2} \sum_{v\in T} (d_S(v) - \bar{d}_S(T))^2 + \frac{1}{4} |\upT| \delta \bar{d}_S(T)^2}{m_S \beta} \geq 0.
        $$
    \end{lemma}
    \begin{proof}
    Since $H$ is a $(\beta, (1-\epsilon)\beta)$-EDCS, we get from the first condition of EDCS that
    \begin{equation}\label{eq:sumdeg-ub}
        \sum_{(u, v) \in H_S} d(u) + d(v) \leq \sum_{(u, v) \in H_S} \beta = m_S \beta.
    \end{equation}
    Let us now focus on the LHS of \cref{eq:sumdeg-ub}. Each vertex $v \in S \cup T$ participates in the sum $d_S(v)$ times for each of its edges in $H_S$, and each time adds a value of $d(v)$ to the sum. Hence,
    \begin{flalign*}
       \sum_{(u, v) \in H_S} d(u) + d(v) &= \sum_{v \in S \cup T} d_S(v) d(v)\\
        &= \sum_{v \in S} d(v)^2 + \sum_{v \in T} d_S(v)d(v) \tag{Since $d_S(v) = d(v)$ for all $v \in S$.}\\
        &= |S| \bar{d}(S)^2 + \sum_{v \in S} (d(v) - \bar{d}(S))^2 + \sum_{v \in T} d_S(v)d(v)
        \tag{This follows from applying \Cref{fact:quadratic-sum} on the first quadratic sum.}\\[0.2cm]
        &\geq m_S \bar{d}(S) + \sum_{v \in S} (d(v) - \bar{d}(S))^2 + \sum_{v \in T} d_S(v)d(v). \tag{Since $|S| \bar{d}(S) = m_S$.}
    \end{flalign*}
    Plugging this lower bound on the LHS of \cref{eq:sumdeg-ub} back to \cref{eq:sumdeg-ub} and moving the terms, we get that
    \begin{equation}\label{eq:gckk-289}
        \sum_{v \in T} d_S(v)d(v) \leq m_S (\beta - \bar{d}(S)) - \sum_{v \in S} (d(v) - \bar{d}(S))^2.
    \end{equation}
    Next, we focus on the LHS of \cref{eq:gckk-289}. For any vertex $v \in T$ define $x_v := d_S(v) - \bar{d}_S(T)$. We have
    \begin{flalign*}
        \sum_{v \in T} d_S(v) d(v) &\geq \sum_{v \in T \setminus \upT} d_S(v)^2 + (1+\alpha)\sum_{v \in \upT} d_S(v)^2 \tag{By definition of $\upT$.}\\
        &= \sum_{v \in T} d_S(v)^2 + \alpha\sum_{v \in \upT} d_S(v)^2\\
        &= \sum_{v \in T} (\bar{d}_S(T) + x_v)^2 + \alpha \sum_{v \in \upT} (\bar{d}_S(T) + x_v)^2 \tag{By definition of $x_v$.}\\
        &= |T| \bar{d}_S(T)^2 + \sum_{v \in T} x_v^2 + 2\bar{d}_S(T) \sum_{v \in T} x_v + \alpha \sum_{v \in \upT} (\bar{d}_S(T) + x_v)^2\\
        &= |T| \bar{d}_S(T)^2 + \sum_{v \in T} x_v^2 + \alpha \sum_{v \in \upT} (\bar{d}_S(T) + x_v)^2 \tag{Since $\sum_{v \in T} x_v = \sum_{v \in T} (d_S(v) - \bar{d}_S(T)) = \sum_{v\in T} d_S(v) - \sum_{v \in T} \bar{d}_S(T) = m_S - m_S = 0$.}\\[0.2cm]
        &\geq |T| \bar{d}_S(T)^2 + (1-\alpha) \sum_{v \in T} x^2_v + \sum_{v \in \upT} (\alpha x_v^2 + \alpha(\bar{d}_S(T) + x_v)^2)\\
        &\geq |T| \bar{d}_S(T)^2 + (1-\alpha) \sum_{v \in T} x^2_v + \sum_{v \in \upT} \alpha \bar{d}_S(T)^2/4 \tag{If $|x_v| \geq \bar{d}_S(T)/2$ then $\alpha x_v^2 \geq \alpha \bar{d}_S(T)^2/4$ otherwise $\alpha(\bar{d}_S(T) + x_v)^2 \geq \alpha \bar{d}_S(T)^2/4$.}\\
        &\geq m_S \cdot \bar{d}_S(T) + \frac{1}{2} \sum_{v\in T} (d_S(v) - \bar{d}_S(T))^2 + |\upT| \cdot \alpha \bar{d}_S(T)^2/4. \tag{By definition of $x_v$, and since $\alpha = 3\sqrt{\delta} < 3 \sqrt{1/60} < 1/2$.}
    \end{flalign*}
    Plugging this lower bound on the LHS of \cref{eq:gckk-289} back to \cref{eq:gckk-289}, we get that
    $$
        m_S \cdot \bar{d}_S(T) + (1-\alpha)\sum_{v\in T} (d_S(v) - \bar{d}_S(T))^2 + |\upT| \cdot \alpha \bar{d}_S(T)^2/4 \leq m_S (\beta - \bar{d}(S)) - \sum_{v \in S}(d(v) - \bar{d}(S))^2.
    $$
    Moving the terms, we get that
    \begin{flalign*}
        \bar{d}_S(T) &\leq (\beta - \bar{d}(S)) - \frac{\sum_{v \in S} (d(v) - \bar{d}(S))^2 + \frac{1}{2} \sum_{v\in T} (d_S(v) - \bar{d}_S(T))^2 + |\upT| \alpha \bar{d}_S(T)^2/4}{m_S}\\
        &= (\beta - \bar{d}(S)) - \gamma \beta \tag{By definition of $\gamma$ in the claim statement.}\\
        &= (1 - \gamma)\beta - \bar{d}(S).
    \end{flalign*}
    This is the desired upper bound on $\bar{d}_S(T)$. Note also that the non-negativity of $\gamma$ follows from the fact that all the terms in it are non-negative.
    \end{proof}
    
    \subsubsection{Some Auxiliary Claims}\label{sec:xgg-3}
    
    Before proving our main characterization, we prove a few useful claims in this section. Namely, that under $\mu(H) < (\frac{2}{3} + \delta)\mu(G)$, we have
    \begin{itemize}
        \item $(1-\epsilon)\beta/2 \leq \bar{d}(S) \leq (1+3\delta) \beta/2$ (stated as \cref{cl:avgd-S}).
        
        That is, the average degree of $S$ in $H$ should be close to $\beta/2$.
        \item $(1-6\delta)\beta/2 \leq \bar{d}_S(T) \leq (1+\epsilon) \beta/2$ (stated as \cref{cl:avgd-T-inS}).
        
        That is, the average degree of $T$ to $S$ in $H$ should be  close to $\beta/2$.
        \item $(\frac{2}{3} - 2\delta) \mu(G) \leq |S| \leq (\frac{2}{3} + 3\delta)\mu(G)$ (stated as \cref{cl:S-2/3-muG}).
        
        That is, $|S|$ should be close to $\frac{2}{3}\mu(G)$.
        
        \item $\gamma \leq 2\delta$ (stated as \cref{cl:gamma<delta}).
        
        That is, the parameter $\gamma$ of \cref{lem:dT-ub} should be small.
    \end{itemize}
    
    We now state and prove these claims one by one.
    
    \begin{claim}\label{cl:avgd-S}
        If $\mu(H) < (\frac{2}{3} + \delta) \mu(G)$, then  $ (1-\epsilon)\beta/2 \leq \bar{d}(S) \leq (1+3\delta) \beta/2$.
    \end{claim}
    \begin{myproof}
        We first prove the lower bound which in fact holds regardless of the assumption of the claim. Since $H$ has no edges between $A$ and $\Bbar$ by definition, then any edge in $\missingMstar$ must be missing from $H$. Hence, by the second condition of EDCS, for any $(u, v) \in \missingMstar$, $d(u) + d(v) \geq (1-\epsilon)\beta$. Thus $\sum_{(u, v) \in \missingMstar} d(u) + d(v) \geq |\missingMstar| (1-\epsilon)\beta$. The LHS equals $\sum_{v \in S} d(v) = m_S$ since $S = V(\missingMstar)$. Dividing through by $|S| = 2|\missingMstar|$ we obtain $\bar{d}(S) \geq (1-\epsilon)\beta/2$.
        
        For the upper bound, suppose for contradiction that $\bar{d}(S) > (1+3\delta) \beta/2$. We have 
        \begin{flalign*}
        \mu(H) &\geq \frac{2\bar{d}(S)}{2\bar{d}(S) + \bar{d}_S(T)} \mu(G) \tag{By \cref{cl:ratio-to-muG}.}\\
        &\geq \frac{2\bar{d}(S)}{2\bar{d}(S) + \beta - \bar{d}(S)}\mu(G) \tag{Since $\bar{d}_S(T) \leq (1-\gamma) \beta - \bar{d}(S) \leq \beta - \bar{d}(S)$ by \cref{lem:dT-ub}.}\\
        &= \frac{2\bar{d}(S)}{\beta + \bar{d}(S)} \mu(G) > \frac{(1+3\delta)\beta}{\beta + (1+3\delta)\beta/2} \mu(G) = \frac{2+6\delta}{3+3\delta} \mu(G) \geq \left(\frac{2}{3} + \delta \right) \mu(G), 
        \end{flalign*}
        where the last inequality holds for all $0 \leq \delta \leq 1/3$. This contradicts the assumption of the claim, and so the claimed upper bound on $\bar{d}(S)$ must hold.
    \end{myproof}
    
    \begin{claim}\label{cl:avgd-T-inS}
        If $\mu(H) < (\frac{2}{3} + \delta) \mu(G)$, then  $ (1-6\delta)\beta/2 \leq \bar{d}_S(T) \leq (1+\epsilon) \beta/2$.
    \end{claim}
    \begin{myproof}
        The upper bound follows from $\bar{d}_S(T) \leq \beta - \bar{d}(S)$ of \cref{lem:dT-ub} and $\bar{d}(S) \geq (1-\epsilon)\beta/2$ of \cref{cl:avgd-S}. For the lower bound, suppose for contradiction that $\bar{d}_S(T) < (1-6\delta)\beta/2$. We have
        \begin{flalign*}
            \mu(H) 
            &\stackrel{\text{\cref{cl:ratio-to-muG}}}{\geq}
            \frac{2\bar{d}(S)}{2\bar{d}(S) + \bar{d}_S(T)} \mu(G) 
            \stackrel{\text{\cref{cl:avgd-S}}}{\geq}
            \frac{(1-\epsilon)\beta}{(1-\epsilon)\beta + \bar{d}_S(T)} \mu(G) \geq \frac{(1-\epsilon)\beta}{\beta + (1-6\delta)\beta/2} \mu(G)\\
            &= \frac{2-2\epsilon}{3 - 6\delta} \mu(G) \geq \left(\frac{2}{3}+ \delta \right) \mu(G),
        \end{flalign*}
        where the last inequality holds for all $1/2 > \delta \geq 2\epsilon \geq 0$. This contradicts the assumption of the claim and proves the lower bound on $\bar{d}_S(T)$.
    \end{myproof}
    
    \begin{claim}\label{cl:S-2/3-muG}
        If $\mu(H) < (\frac{2}{3} + \delta) \mu(G)$, then $(\frac{2}{3} - 2\delta) \mu(G) \leq |S| \leq (\frac{2}{3} + 3\delta)\mu(G)$.
    \end{claim}
    \begin{myproof}
            For the lower bound, observe that
            $$
                |S|
                \stackrel{\text{\cref{cl:S-lb}}}{\geq}
                2(\mu(G) - \mu(H)) > 2\Big(\mu(G) - (\frac{2}{3} + \delta) \mu(G)\Big) \geq \Big(\frac{2}{3} - 2\delta \Big)\mu(G).
            $$
            For the upper bound, we have
            $$
                \mu(H)
                \stackrel{\text{\cref{cl:muH=T}}}{=}
                |T| = \frac{m_S}{\bar{d}_S(T)} = \frac{|S|\bar{d}(S)}{\bar{d}_S(T)}
                \stackrel{\text{\cref{cl:avgd-S,cl:avgd-T-inS}}}{\geq}
                \frac{|S| (1-\epsilon)\beta/2}{(1+\epsilon)\beta/2} \geq (1-2\epsilon)|S|.
            $$
            Hence, the assumption $\mu(H) < (\frac{2}{3} + \delta) \mu(G)$ implies that 
            $$
            |S| \leq \frac{1}{1-2\epsilon}(\frac{2}{3} + \delta) \mu(G) 
            \stackrel{(\epsilon < 1/120)}{<}
            (1+4\epsilon)(\frac{2}{3} + \delta)\mu(G)
            \stackrel{(\delta > 2\epsilon, 0 < \delta < 1/60)}{<}
            (\frac{2}{3} + 3\delta) \mu(G).
            $$
    \end{myproof}
    
    \begin{claim}\label{cl:gamma<delta}
        If $\mu(H) < (\frac{2}{3} + \delta) \mu(G)$, then $\gamma \leq 2\delta$.
    \end{claim}
    \begin{myproof}
    Suppose for contradiction that $\gamma > 2\delta$. Then
    \begin{flalign*}
        \bar{d}_S(T) &\leq  (1-\gamma)\beta - \bar{d}(S) \tag{By \cref{lem:dT-ub}}\\
        &< (1-2\delta)\beta - \bar{d}(S)\\
        &\leq (1-2\delta)\beta - (1+3\delta)\beta/2 \tag{By \cref{cl:avgd-S}}\\
        &\leq (1-6\delta)\beta/2.
    \end{flalign*}
    This contradicts \cref{cl:avgd-T-inS} that $\bar{d}_S(T) \geq (1-6\delta)\beta/2$ and proves the claim.
    \end{myproof}
    
    \subsubsection{The Main Characterization}\label{sec:xgg-4}
    
    Having proved the auxiliary claims above and the parametrized guarantee of \cref{lem:dT-ub}, we ready to prove the main characterizations of this section on $|S \setminus \vmid|$, $|T \setminus \vmid$, and $|W \setminus \vlow|$.
    
    \begin{claim}\label{cl:S/mid-small}
        If $\mu(H) < (\frac{2}{3} + \delta) \mu(G)$, then $|S \setminus \vmid| \leq 199 \cdot \delta \mu(G)$.
    \end{claim}
    \begin{myproof}
    Suppose for contradiction that $|S \setminus \vmid| > 199 \cdot \delta \mu(G)$. If $v \not\in \vmid$, then $d(v) \not\in [.4\beta, .6\beta]$ by definition of $\vmid$. Since $ (1-\epsilon)\beta/2 \leq \bar{d}(S) \leq (1+3\delta) \beta/2$ by \cref{cl:avgd-S} and $\delta < 1/60$ and $\epsilon < 1/120$, we get $.49\beta < \bar{d}_S(T) < .53 \beta$. Thus for all $v \in S \setminus \vmid$, we have $(d(v) - \bar{d}(S))^2 \geq (.07\beta)^2$. Hence, by definition of $\gamma$ in \cref{lem:dT-ub} and the non-negativity of the terms in its numerator, we get that
    \begin{flalign*}
        \gamma &\geq \frac{\sum_{v \in S} (d(v) - \bar{d}(S))^2}{m_S \beta} \geq \frac{|S \setminus \vmid| \cdot (.07 \beta)^2}{m_S \beta} = 
        \frac{|S \setminus \vmid| \cdot (.07 \beta)^2}{|S| \bar{d}(S) \beta}
        \stackrel{\text{\cref{cl:avgd-S}}}{\geq}
        \frac{|S \setminus \vmid| \cdot (.07 \beta)^2}{|S| (1+3\delta) \beta^2/2}\\
        &\geq 2(1-4\delta)(0.07)^2 \frac{|S \setminus \vmid|}{|S|} \geq \frac{1}{110} \cdot \frac{|S \setminus \vmid|}{|S|}. \tag{Since $\delta < 1/60$.}\\
        &\geq \frac{1}{110} \cdot \frac{199 \cdot \delta \mu(G)}{(\frac{2}{3}+3\delta)\mu(G)} \tag{By \cref{cl:S-2/3-muG} and the assumption that $|S \setminus \vmid| > 199 \cdot \delta \mu(G)$.}\\
        &> 2 \delta. \tag{Since $\delta < 1/60$.}
    \end{flalign*}
    But this contradicts \cref{cl:gamma<delta} that $\gamma \leq 2\delta$, completing the proof.
    \end{myproof}
    
    \begin{claim}\label{cl:T/mid-small}
        If $\mu(H) < (\frac{2}{3} + \delta) \mu(G)$, then $|T \setminus \vmid| \leq 600 \cdot \delta \mu(G)$.
    \end{claim}
    \begin{myproof}
    Suppose for contradiction that $|T \setminus \vmid| > 600 \cdot \delta \mu(G)$. If $v \not\in \vmid$, then by definition of $\vmid$, we have $d(v) \not\in [.4\beta, .6\beta]$. Since by \cref{cl:avgd-T-inS} $ (1-6\delta)\beta/2 \leq \bar{d}_S(T) \leq (1+\epsilon) \beta/2$ and $\delta < 1/60$ and $\epsilon < 1/120$, we get $.45\beta \leq \bar{d}_S(T) < .51 \beta$. Thus for all $v \in T \setminus \vmid$, we have $(d_S(v) - \bar{d}_S(T))^2 \geq (.05 \beta)^2$. Hence, by definition of $\gamma$ in \cref{lem:dT-ub} and the non-negativity of the terms in its numerator,
    \begin{flalign*}
        \gamma &\geq \frac{\frac{1}{2} \sum_{v\in T} (d_S(v) - \bar{d}_S(T))^2}{m_S \beta} \geq
        \frac{|T \setminus \vmid| \cdot (.05 \beta)^2}{2m_S \beta} = 
        \frac{|T \setminus \vmid| \cdot (.05 \beta)^2}{2|S| \bar{d}(S) \beta}
        \stackrel{\text{\cref{cl:avgd-S}}}{\geq}
        \frac{|T \setminus \vmid| \cdot (.05 \beta)^2}{|S| (1+3\delta) \beta^2}\\
        &\geq (1-4\delta) (0.05)^2 \frac{|T \setminus \vmid|}{|S|} \geq \frac{1}{414} \cdot \frac{|T \setminus \vmid|}{|S|} \tag{Since $\delta < 1/60$.}\\
        &\geq \frac{1}{414} \cdot \frac{600 \delta \mu(G)}{(\frac{2}{3} + 3\delta)\mu(G)} \tag{By \cref{cl:S-2/3-muG} and the assumption that $|T \setminus \vmid| > 600 \delta \mu(G)$.} \\
        &> 2\delta. \tag{Since $\delta < 1/60$.}
    \end{flalign*}
    But this contradicts \cref{cl:gamma<delta} that $\gamma \leq 2\delta$, completing the proof.
    \end{myproof}
    
    \begin{claim}\label{cl:WcapU-small}
        If $\mu(H) < (\frac{2}{3} + \delta) \mu(G)$, then $|W \setminus \vlow| \leq 33 \sqrt{\delta} \mu(G)$.
    \end{claim}
    \begin{myproof}
        First, we claim that $|\upT| < 22\delta|S|/\alpha$. Suppose for  contradiction that $|\upT| \geq 22\delta|S|/\alpha$. By definition of $\gamma$ in \cref{lem:dT-ub} and the non-negativity of the terms in its numerator, we have
        \begin{flalign*}
            \gamma &\geq \frac{\frac{1}{4}|\upT| \alpha \bar{d}_S(T)^2}{m_S \beta}
            = 
            \frac{|\upT| \alpha \bar{d}_S(T)^2}{4 |S| \bar{d}(S) \beta}
            \stackrel{\text{\cref{cl:avgd-S,cl:avgd-T-inS}}}{\geq}
            \frac{|\upT| \alpha ((1-6\delta)\beta/2)^2}{4|S|((1+3\delta)\beta/2)\beta}
            =
            \frac{1}{8} \cdot \frac{\alpha (1-6\delta)^2}{(1+3\delta)} \cdot \frac{|\upT|}{|S|}\\
            &\geq \frac{1}{8} \alpha (1-15\delta) \cdot \frac{|\upT|}{|S|}
            \stackrel{(0 < \delta < 1/60)}{\geq}
            \frac{\alpha}{11} \cdot  \frac{\upT}{|S|} \stackrel{(|\upT| \geq 22\delta|S|/\alpha)}{\geq} 2\delta.
        \end{flalign*}
        But this contradicts \cref{cl:gamma<delta} that $\gamma \leq 2\delta$, therefore we must have $|\upT| < 22\delta|S|/\alpha$.

        Now take a vertex $v \in W \setminus \vlow$. We have $d(v) \geq .2\beta$ by definition of $\vlow$. This means that $m_W \geq .2\beta |W \setminus \vlow|$. On the other hand, since any vertex has degree at most $\beta$ in a $(\beta, (1-\epsilon)\beta)$-EDCS, and for any vertex $v \in T \setminus \upT$ we have $d(v) \leq (1+\alpha)d_S(v)$ by definition of $\upT$, we get
        $$
            m_W = \sum_{v \in T} d_W(v) = \sum_{v \in T} (d(v) - d_S(v)) \leq \sum_{v \in T \setminus \upT} \alpha d_S(v) + \sum_{v \in \upT} \beta \leq \alpha m_S + \beta|\upT|.
        $$
        Additionally, since $m_S = |S| \bar{d}(S)$ we get by \cref{cl:avgd-S} that $m_S \leq |S| (1+3\delta)\beta/2$. Combined with our earlier lower bound on $m_W$, this implies that
        $$
            .2 \beta |W \setminus \vlow| \leq \alpha |S|(1+3\delta)\beta/2 + \beta |\upT| \qquad \Leftrightarrow \qquad |W \setminus \vlow| \leq 5 (\alpha (1+3\delta)|S|/2 + |\upT|).
        $$
        Using our earlier bound of $|\upT| < 22\delta|S|/\alpha$, we get that
        $$
            |W \setminus \vlow| \leq 5(\alpha(1+3\delta)|S|/2 + 22\delta|S|/\alpha).
        $$
        Since $\alpha = 3\sqrt{\delta}$ by definition and $0 < \delta < 1/60$, this implies $|W \setminus \vlow| < 46 \sqrt{\delta}|S|$. Given the upper bound of $|S| \leq (\frac{2}{3} + 3\delta)\mu(G)$ in \cref{cl:S-2/3-muG} and since $\delta < 1/60$, we get $|W \setminus \vlow| \leq 46(\frac{2}{3}+3\delta) \mu(G) \leq 33 \sqrt{\delta}\mu(G)$.
    \end{myproof}
    
    As discussed earlier, \cref{cl:S/mid-small,cl:T/mid-small,cl:WcapU-small} together imply \cref{lem:EDCS-char} for bipartite graphs.
\section{Beating Half for General Graphs}\label{sec:general}

Our discussion of \cref{sec:bipartite} crucially relied on the graph $G$ being bipartite. In this section, we prove \cref{thm:main} for general graphs. Our main result of this section is the following semi-dynamic algorithm, akin to \cref{lem:A-for-bipartite}, but now for general graphs.

\begin{lemma}\label{lem:A-for-general-1/2}
    For any $\epsilon > 0$ and any $n$-vertex fully dynamic graph $G$, there is a (randomized) data structure $\mc{A}$ that takes $(\poly\log n)$ worst-case update-time, and upon being queried takes $\widetilde{O}(n/\epsilon^5)$ time to produce a number $\widetilde{\mu}$ such that 
    $
        .5018 \cdot \mu(G) - \epsilon n \leq \E[\widetilde{\mu}] \leq \mu(G).
    $
\end{lemma}

The proof of \cref{thm:main} for general graphs follows from \cref{lem:A-for-general-1/2}:

\begin{proof}[Proof of \cref{thm:main} for general graphs]
    Follows by plugging the data structure $\mc{A}$ of \cref{lem:A-for-general-1/2} as the data structure $\mc{A}$ in \cref{lem:A-to-B} and choosing sufficiently small $\epsilon$ such that $.501 \geq .5018 - \epsilon$.
\end{proof}

\subsection{The Semi-Dynamic Algorithm (Proof of \cref{lem:A-for-general-1/2})}

The data structures that we maintain are exactly the same as those in \cref{sec:bipartite}. Namely, we maintain the adjacency matrix of the graph $G$ and a $(1/2-\epsilon)$-approximate matching $M$ of $G$ in $(\poly\log n)$ worst-case update-time \cite{BernsteinFH-SODA19,BehnezhadDHSS-FOCS19} against adaptive adversaries. It remains to show how to produce the number $\widetilde{\mu}$ in $\widetilde{O}(n/\epsilon^5)$ time using these data structures, which is what we focus on in the rest of this section.

The first idea is to define a (random) bipartite subgraph $G_B=(L, R, E_B)$ of $G$. We construct $G_B$ in a way that all edges of $M$ belong to $E_B$. Specifically, for each edge in $M$ we put one of its endpoints in $L$ and the other in $R$ arbitrarily. The rest of the vertices (i.e., $V \setminus V(M)$) are independently and uniformly added either to $L$ or $R$. An edge belongs to $G_B$ iff it belongs to $G$ and it has one endpoint in $L$ and one in $R$. We note that a similar randomization was used in \cite{BehnezhadLM-SODA20}.

\begin{observation}
    $M$ is a maximal matching of $G_B$.
\end{observation}
\begin{myproof}
    Holds since $M \subseteq E_B$, $M$ is a maximal matching of $G$, and $G_B$ is a subgraph of $G$.
\end{myproof}

Note that if we had $\mu(G_B) = \mu(G)$, we could simply run the algorithm of \cref{sec:bipartite} on graph $G_B$. However, $\mu(G_B)$ can be smaller than $\mu(G)$, and so additional ideas are needed.

The following \cref{alg:static-gl} is analogous to \cref{alg:static} of \cref{sec:bipartite}. 

\begin{alglist}{}{alg:static-gl}
\begin{enumerate}[leftmargin=15pt, itemsep=0pt]
    \item Let $M$ and $G_B=(L, R, E_B)$ be as above.
    
    \item Let $M' \subseteq M$ include each edge of $M$ independently with probability $p = .03$.
    
    \item Let $V' := V(M')$ and $U := V \setminus V(M)$.
    
    \item Let $V'_R := V' \cap R$, $V'_L := V' \cap L$, $U_R := U \cap R$, $U_L := U \cap L$.
    
    \item Let $H_R := G_B[V'_R, U_L]$ be the induced bipartite subgraph of $G_B$ between $V'_R$ and $U_L$.
    
    \item Let $H_L := G_B[V'_L, U_R]$ be the induced bipartite subgraph of $G_B$ between $V'_L$ and $U_R$.
    
    \item For $e = (u, v)$, $v \in R$, $u \in L$, let $q_e := \Pr_\pi[v \in \GMM{H_R, \pi}, u \in \GMM{H_L, \pi} \mid M']$ for a random permutation $\pi$.
    
    \item Let $\ell := |\GMM{G[V \setminus V(M)], \pi}|$ for any arbitrary permutation $\pi$.

    \item Return $\widetilde{\mu}' := |M| + \max\{\ell, \sum_{e \in M'} q_e\}$.
\end{enumerate}
\end{alglist}

To implement \cref{alg:static-gl}, we use the following \cref{prop:given-vertex}, which builds on the techniques developed in \cite{behnezhad2021}. See \cref{sec:sublinear-apx} for the proof.

\newcommand{\propGiven}[0]{Let $G=(V, E)$ be an $n$-vertex graph to which we have adjacency matrix query access and let $K \subseteq V$ be an arbitrary subset. For any $\epsilon > 0$ and $v \in K$ chosen u.a.r., there is an algorithm that succeeds with probability $1 - \epsilon n /|K|$ and in $\widetilde{O}(n^2/(\epsilon |K|))$ expected time returns whether $v$ is matched by $\GMM{G, \pi}$, where $\pi$ is a u.a.r. permutation of $E$ drawn by the algorithm. The probabilistic statements depend both on the randomization of $\pi$ and the randomization of $v \sim K$.}

\begin{proposition}[\cite{behnezhad2021}]\label{prop:given-vertex}
    \propGiven{}
\end{proposition}

\begin{lemma}\label{lem:implementation-gl}
    For any $\epsilon > 0$, there is an algorithm that w.h.p. takes $\widetilde{O}(n/\epsilon^5)$ time and returns a number $\widetilde{\mu}''$ such that $\widetilde{\mu}' - \epsilon n \leq \widetilde{\mu}'' \leq \widetilde{\mu}'$. Here $\widetilde{\mu}'$ is the output of  \cref{alg:static-gl}.
\end{lemma}
\begin{myproof}
    Since $M$ is given, we can construct $L$, $R$, $M'$, $V'_R$, $V'_L$, $U_R$, and $U_L$ in $O(n)$ time and also store for each vertex to which one of these sets it belongs.  However, we will not compute $H_R$ or $H_L$ explicitly as they may have $\Omega(n^2)$ edges. Note, however, that any adjacency matrix query to $H_R$ or $H_L$ can be answered in $O(1)$ time since we have adjacency matrix access to $G$ and explicitly have these graphs' vertex sets stored.
    
    Let us condition on the outcome of $M'$ for the rest of the proof. Observe that $\sum_{e \in M'} q_e \leq \sum_{e \in M'} 1 = |M'|$. So if $|M'| \leq \epsilon n$ then returning $\widetilde{\mu}'' = |M|$ proves the lemma. Thus, let us assume $|M'| > \epsilon n$. We do not know how to compute $q_e$ for every edge in $M'$. Instead, we show how to estimate the value of the sum $\sum_{e \in M'} q_e$. 
    
    Let $k = 48\log n / \epsilon^2$. For any $i \in [k]$, we pick an edge $e_i = (v_i, u_i)$ from $M'$ each uniformly at random (with replacement), assuming w.l.o.g. that $u_i \in L, v_i \in R$. For any $i \in [k]$, we run \cref{prop:given-vertex} once on graph $H_R$ for vertex $v_i$ and once on graph $H_L$ for vertex $u_i$ for error parameter $\epsilon' = \epsilon^2/2$. We then let $X_i$ be the indicator of the event that both $v_i$ and $u_i$ are returned to be matched by \cref{prop:given-vertex}. Since $H_R$ and $H_L$ are vertex disjoint by construction, the dependence of $u_i$ and $v_i$ (in that they are both endpoints of the same edge in $M'$) does not affect the guarantees of \cref{prop:given-vertex}. In particular, $u_i$ (resp. $v_i$) is still a vertex chosen u.a.r. from $V'_L$ (resp. $V'_R$). 
    
    Since $|V'_R| = |V'_L| = |M'| \geq \epsilon n$, the set $K$ in our call to \cref{prop:given-vertex} has size $\epsilon n$ at least. Hence, \cref{prop:given-vertex} takes $\widetilde{O}(n^2/(\epsilon' \epsilon n)) = \widetilde{O}(n/\epsilon^3)$ expected time for each $i \in [k]$, and has success probability $1- \epsilon' n /(\epsilon n) = 1-\epsilon/2$. Since we call it $k$ times, the total time-complexity is $\widetilde{O}(n k / \epsilon^2) = \widetilde{O}(n /\epsilon^5)$ in expectation. We will show later how to turn this into high probability.
    
    Since, as discussed, our call to \cref{prop:given-vertex} has failure probability $\leq \epsilon/2$, we get that
    \begin{equation}\label{eq:Mlrg-20}
        \E[X_i] = \frac{1}{|M'|} \sum_{(u, v) \in M'} (q_e \pm \epsilon/2) = \left(\frac{1}{|M'|} \sum_{(u, v) \in M'} q_e \right) \pm \epsilon/2.
    \end{equation}
    Define $X := \sum_i X_i$ and $Q := \frac{X|M'|}{k}$. Since the $X_i$'s are independent (as each call to \cref{prop:given-vertex} generates a fresh random permutation), we get from the Chernoff bound that with probability $1-2n^{-4}$, $X = \E[X] \pm \sqrt{12\E[X] \log n}$. As such, we get that w.h.p.
    \begin{flalign*}
        Q &= \frac{(\E[X] \pm \sqrt{12\E[X] \log n})|M'|}{k} = \frac{(k\E[X_i] \pm \sqrt{12 k\E[X_i] \log n})|M'|}{k}\\
        &= \E[X_i] |M'| \pm .5 \epsilon |M'| \tag{Since $\E[X_i] \leq 1$ and $k = 48 \log n / \epsilon^2$.}\\
        &= \sum_{(u, v) \in M'} q_e \pm \epsilon|M'| \tag{By \cref{eq:Mlrg-20}.}\\
        &= \sum_{(u, v) \in M'} q_e \pm .5 \epsilon n. \tag{Since $|M'| \leq n/2$.}
    \end{flalign*}
    Note also that $\ell$ is simple to approximate within a $(1+\epsilon)$ factor using the algorithm of \cite{behnezhad2021} as black-box. Therefore, returning $\widetilde{\mu}'' = |M| + \max\{\ell, Q\} - .5\epsilon n$ guarantees $\widetilde{\mu}' - \epsilon n \leq \widetilde{\mu}'' \leq \widetilde{\mu}'$ w.h.p.
    
    Since the expected running time is $\widetilde{O}(n/\epsilon^5)$, by Markov's inequality the algorithm terminates in $2 \times \widetilde{O}(n/\epsilon^5)$ time with probability at least $1/2$. Thus, we can run $O(\log n)$ independent instances of the algorithm, and return the output of the one that first terminates. This way, our algorithm w.h.p. terminates in $\widetilde{O}(n/\epsilon^5)$ time. Since the guarantee on $\widetilde{\mu}''$ holds with probability $1-1/\poly(n)$, it should hold for all $O(\log n)$ instances (and so the one that first terminates) still w.h.p.
\end{myproof}

Next, we turn to analyze the approximation ratio of \cref{alg:static-gl}. 

\cref{prop:greedy-random-vertex}, which was used in the proof of \cref{lem:mu-lb}, only gives a lower bound on the size of the matching. For our discussion of this section, however, we need a more fine-tuned bound guaranteed by the following proposition of  \cite{BehnezhadLM-SODA20}.

\begin{proposition}[{\cite[Lemma~5.2]{BehnezhadLM-SODA20}}]\label{prop:BLM}
    Let $0 < p \leq 1$, let $G=(A, B, E)$ be a bipartite graph, let $A' \subseteq A$ include each vertex of $A$ independently with probability $p$, and let $H$ be the induced subgraph of $G$ on vertex-set $A' \cup B$. Fix an arbitrary permutation $\pi$ over the edge-set of $H$ and fix an arbitrary matching $M$ of $G$. Let $X$ be the number of edges in $M$ whose endpoint in $A$ is matched in $\GMM{H, \pi}$; then
    $$
    \E_{A'}[X] \geq p(|M| - 2p|A|).
    $$
\end{proposition}

\begin{lemma}\label{lem:mu-lb-gl}
    For \cref{alg:static-gl}, it holds that $\E[\widetilde{\mu}'] \geq .5018 \mu(G)$.
\end{lemma}
\begin{myproof}
    Fix an arbitrary maximum matching $M^\star$ of $G$. Observe that there are exactly $\mu(G) - |M^\star|$ augmenting paths in $M \oplus M^\star$ for $M$. Denoting the number of length one augmenting paths in $M \oplus M^\star$ by $L_1$, there are exactly $\mu(G) - |M| - L_1$ augmenting paths of length at least three. Each of these augmenting paths has exactly two (endpoint) edges that have one vertex matched in $M$ and one endpoint unmatched in $M$. Putting together these edges, we obtain a matching $Z$ with $2(\mu(G) - |M| - L_1)$ edges, all of which belong to $M^\star$. The number of length two components in $M \oplus Z$ is at most $|M|$. The rest of the components are length three augmenting paths for $M$. Thus, there are at least $|Z| - |M| = 2\mu(G) - 3|M| - 2L_1$ length three augmenting paths for $M$ in $M \oplus Z$. Let $P=(a, u, v, b)$ be one of these leng-three augmenting paths with $(u, v) \in M$ and $(a, u), (v, b) \in M^\star$. Suppose w.l.o.g. that we assign $u \in L$ and $v \in R$ in \cref{alg:static-gl}. Then $P$ remains an augmenting path in $G_B$ if $a \in R$ and $b \in L$, which happens with probability $1/4$. Under this event, we say $P$ {\em survives} to $G_B$. Let $\mc{P}$ be the set of all the length three augmenting paths in $M^\star \oplus M$ that survive to $G_B$, and note that
    \begin{equation}\label{eq:Mlrcu881}
        \E_{L, R}|\mc{P}| \geq \frac{1}{4} (2\mu(G) - 3|M| - 2L_1).
    \end{equation} 
    Define
    \begin{flalign*}
        M^\star_L &:= \{ (a, u) \in M^\star \mid u \in L, a \in R, (a, u, \cdot, \cdot) \in \mc{P} \},\\
        M^\star_R &:= \{ (v, u) \in M^\star \mid v \in R, b \in L, (\cdot, \cdot, v, b) \in \mc{P} \},\\
        M_{\mc{P}} &:= \{(u, v) \in M \mid (\cdot, u, v, \cdot) \in \mc{P} \}.
    \end{flalign*}
    Furthermore, define 
    $$
        F_L := G_B[V(M) \cap L, U_R], \qquad F_R := G_B[V(M) \cap R, U_L].
    $$
    
    Let $\pi$ be any arbitrary permutation of the edges in $E$. Let $X_L$ (resp. $X_R$) denote the number of vertices in $V(M^\star_L) \cap L$ that are matched by $\GMM{H_L, \pi}$ (resp. $\GMM{H_R, \pi}$). Noting that $H_L$ is an induced subgraph of $F_L$ including each of its vertices in its $V(M) \cap L$ part independently from each other with probability $p$, we can apply \cref{prop:BLM} (on graph $F_L$ fixing matching $M^\star_L$) to obtain that
    \begin{equation}\label{eq:gcr-93888}
        \E_{M'}[X_L] \geq p(|M^\star_L| - 2p|V(M) \cap L|) = p(|\mc{P}| - 2p|M|).
    \end{equation}
    With essentially the same proof, we also get that
    \begin{equation}\label{eq:gcr-93888-2}
        \E_{M'}[X_R] \geq p(|M^\star_R| - 2p|V(M) \cap R|) = p(|\mc{P}| - 2p|M|).
    \end{equation}
    
    Now let $M'_\mc{P}$ be the edges in $M_\mc{P}$ that also belong to subsample $M'$ of $M$ in \cref{alg:static-gl}. Let $Y$ be the number of edges $(u, v) \in M'_{\mc{P}}$ where $u$ is matched by $\GMM{H_L, \pi}$ and $v$ is matched by $\GMM{H_R, \pi}$. We have
    $$
    Y \geq |M'_{\mc{P}}| - (|M'_{\mc{P}}| - X_L) - (|M'_{\mc{P}}| - X_R) = X_R + X_L - |M'_{\mc{P}}|.
    $$
    Taking expectation over $M'$, plugging \cref{eq:gcr-93888,eq:gcr-93888-2}, and noting that $\E|M'_\mc{P}| = p|M_\mc{P}| = p|\mc{P}|$, we get
    $$
        \E_{M'}[Y] \geq 2p(|\mc{P}| - 2p|M|) - p|\mc{P}| = p|\mc{P}| - 4p^2|M|.
    $$
    Further taking expectation over the randomization of $L, R$, we get that
    \begin{flalign*}
        \E_{M', L, R}[Y] \geq p \E_{L, R}|\mc{P}| - 4p^2|M| &\stackrel{\cref{eq:Mlrcu881}}{\geq} p \cdot  \frac{1}{4} (2\mu(G)-3|M|-2L_1) - 4p^2|M|\\
        &= \frac{p}{2} \mu(G) - (\frac{3p}{4}+4p^2) |M| - \frac{p}{2}L_1.
    \end{flalign*}
    From this, we get that
    \begin{equation}\label{eq:dgc-21389}
        \E\Big[\sum_{e \in M'} q_e \Big] \geq \E\Big[\sum_{e \in M'_\mc{P}} q_e \Big] = \E[Y] \geq \frac{p}{2} \mu(G) - (\frac{3p}{4}+4p^2) |M| - \frac{p}{2}L_1.
    \end{equation}
    Therefore, we have
    \begin{flalign*}
        \E[\widetilde{\mu}'] &= |M| + \max \left\{\ell, \,\, \E\Big[\sum_{e \in M'} q_e \Big] \right\}\\
        &\geq |M| + \max \left\{\frac{L_1}{2}, \,\, \E\Big[\sum_{e \in M'} q_e \Big] \right\} \tag{Since $\ell$ is the size of a maximal matching in $G[V \setminus V(M)]$ and $\mu(G[V \setminus V(M)]) \geq L_1$.}\\
        &\geq |M| + \max\left\{ \frac{L_1}{2},\,\, \frac{p}{2} \mu(G) - (\frac{3p}{4}+4p^2) |M| - \frac{p}{2}L_1 \right\} \tag{By \cref{eq:dgc-21389}.}\\
        &= \max\left\{ |M| + \frac{L_1}{2} ,\,\, \frac{p}{2}\mu(G) + (1 - \frac{3p}{4}-4p^2)|M| - \frac{p}{2} L_1 \right\}\\
        &\geq (1-\epsilon) \max\left\{ \frac{\mu(G)}{2} + \frac{L_1}{2} ,\,\, (\frac{1}{2} + \frac{1}{8}p - 2p^2) \mu(G) - \frac{p}{2} L_1 \right\} \tag{Since $|M| \geq (1-\epsilon)\mu(G)/2$ and $(1-\frac{3p}{4}-4p^2) = 0.9739 > 0$ as $p=0.03$.}\\
        &\geq (1-\epsilon) \max\left\{ \frac{\mu(G)}{2} + \frac{L_1}{2} ,\,\, .5019 \mu(G) - .015 L_1 \right\} \tag{Since $p = .03$.}\\
        &\geq (1-\epsilon) .5018 \mu(G). \tag{This holds for all values of $L_1$.}
    \end{flalign*}
    This completes the proof.
\end{myproof}

\begin{lemma}\label{lem:mu-ub-gl}
    For \cref{alg:static-gl}, it holds with probability 1 that $\widetilde{\mu}' \leq \mu(G)$.
\end{lemma}
\begin{proof}
    Condition on the outcome of $M'$ in \cref{alg:static-gl}. Then run the process of constructing matchings $\GMM{H_R, \pi}$ and $\GMM{H_L, \pi}$ for a random $\pi$. Define
    $$
        \mc{Y} = \{(a, u, v, b) \mid (a, u) \in \GMM{H_L, \pi}, (u, v) \in M, (v, b) \in \GMM{H_R, \pi}\}.
    $$
    Note that $\mc{Y}$ is a collection of length-three augmenting paths for $M$. Hence, we can apply all of them at the same time on $M$. This implies that
    $
        \mu(G) \geq |M| + |\mc{Y}|.
    $
    Moreover, we have 
    $$
        \E_\pi[|\mc{Y}| \mid M'] = \sum_{e \in M'} \Pr[v \in \GMM{H_R, \pi}, u \in \GMM{H_L, \pi} \mid M'] = \sum_{e \in M'} q_e.
    $$
    Given this expected value, there must be a choice of $\pi$ with $|\mc{Y}| \geq \sum_{e \in M'} q_e$. This suffices to show
    $$
        \mu(G) \geq |M| + \sum_{e \in M'} q_e = \widetilde{\mu}'. \qedhere
    $$
\end{proof}

We are now ready to complete the proof of \cref{lem:A-for-general-1/2}.

\begin{proof}[Proof of \cref{lem:A-for-general-1/2}]
    The data structures that we store, as discussed, take only $(\poly\log n)$ worst-case time to maintain against an adaptive adversary. When the algorithm is queried, we return the output $\widetilde{\mu}''$ of \cref{lem:implementation-gl}. It takes $\widetilde{O}(n/\epsilon^5)$ time to produce this by \cref{lem:implementation-gl}, which is the desired query time of \cref{lem:A-for-general-1/2}. Moreover, for the approximation ratio, we have
    $$
        .5018 \mu(G) - \epsilon n
        \stackrel{\text{\cref{lem:mu-lb-gl}}}{\leq}
        \E[\widetilde{\mu}'] - \epsilon n
        \stackrel{\text{\cref{lem:implementation-gl}}}{\leq} 
        \E[\widetilde{\mu}''] \stackrel{\text{\cref{lem:implementation-gl}}}{\leq} 
        \E[\widetilde{\mu}']
        \stackrel{\text{\cref{lem:mu-ub-gl}}}{\leq}
        \mu(G).
    $$
    This completes the proof of \cref{lem:A-for-general-1/2}.
\end{proof}

\bibliographystyle{plainnat}
\bibliography{references}

\appendix
\section{Needed Sublinear Algorithms From \cite{behnezhad2021}}\label{sec:sublinear-apx}

In our proofs, we used \cref{prop:sublinear,prop:given-vertex} which are implied by a result of the author in \cite{behnezhad2021}. In this section, we prove why these propositions follow from \cite{behnezhad2021}.

In its Section~4, \cite{behnezhad2021} gives an algorithm that works in the adjacency list model. The algorithm is then adapted to the adjacency matrix model using a reduction that is provided in \cite[Section~5]{behnezhad2021}. The reduction works as follows. Let $G=(V, E)$ be the graph to which we have adjacency matrix access. A graph $H=(V_H, E_H)$ is then defined based on $G$ in such a way that any adjacency list query to $H$ can be answered with a single adjacency matrix query to $G$. The graph $H$ has two vertex disjoint copies $G_1=(V_1, E_1)$ and $G_2=(V_2, E_2)$ of  $G$ with a number of edges added between $G_1$ and $G_2$. Additionally each vertex in $V_2$ has $10n/\epsilon$ leaves adjacent to it where $n = |V|$. 

Partition $V_2$ into $V'_2$ and $V''_2 := V_2 \setminus V'_2$ such that $V'_2$ includes a vertex $v \in V_2$ iff the lowest rank edge of $v$ according to $\pi$ is to a vertex in $V_1 \cup V_2$ (i.e., not to a leaf of $v$). Let $A := \GMM{H[V_1], \pi} = \GMM{G, \pi}$, $B := \GMM{H, \pi} \cap ((V_1 \times V_1) \cup (V_1 \times V'_2))$, and $C = B \cap (V_1 \times V_1)$.

\begin{observation}\label{obs:mmlrg-2103}
    Any vertex in $V$ belongs to $V'_2$ with probability at most $\epsilon$.
\end{observation}
\begin{myproof}
    For a random permutation $\pi$ over $E_H$, the lowest rank edge of each vertex $v \in V_2$ goes to $V_1 \cup V_2$ with probability at most $\frac{2n}{10n/\epsilon + 2n} < \epsilon$.
\end{myproof}

\vspace{-0.2cm}
Define the {\em match-status} of a vertex $v$ in a matching $M$ to be the indicator $\pmb{1}(v \in V(M))$.

\vspace{-0.2cm}
\begin{claim}\label{cl:match-statusAC}
    For a random $\pi$, there are, in expectation, at most $3\epsilon n$ vertices in $V_1$ whose match-status for $A$ is different from $C$.
\end{claim}
\begin{myproof}
    Define $S = A \oplus B$. We claim that any component of $S$ must be a path with at least one endpoint in $V'_2$. To see this, take an arbitrary component $C$ of $S$. The lowest rank edge $e \in C$ must have an endpoint in $V'_2$ or else we should have $e \in A, B$ as $e$ belongs to both $H[V_1]$ and $H$ and has no lower rank edge in either $A$ or $B$, which contradicts $e \in S$. This edge $e$ should be an endpoint of the path $C$, since any vertex in $V'_2$ has degree at most 1 in $S$ (as it cannot be matched in $A$). From this, we get that the total number of connected components in $S$ is at most $|V'_2|$. As a result, there are at most $2|V'_2|$ vertices with a different match-status in $A$ and $B$.
    
    Since $C$ is a sub-matching of $B$, excluding its edges that go from $V_1$ to $V'_2$, any vertex $v$ with a different match-status in $A$ and $C$ must either have a different match-status in $A$ and $B$, or it should be matched in both $A$ and $B$, but its match in $B$ is through a $V_1 \times V'_2$ edge. So, in total, at most $2|V'_2| + |V'_2|$ vertices may have a different match-status in $A$ and $C$. This completes the proof since $\E|V'_2| \leq \epsilon n$ by \cref{obs:mmlrg-2103}.
\end{myproof}

\vspace{-0.5cm}
\begin{proof}[Proof of \cref{prop:sublinear}]
It is shown in the proof of \cite[Lemma~5.4]{behnezhad2021} that one can w.h.p. find an estimate of $\E|C|$ (denoted $M_1(\pi)$ in \cite{behnezhad2021}) with an additive error of $O(\epsilon n)$ in $\widetilde{O}(n/\epsilon^3)$ time. \cref{cl:match-statusAC} above shows that $\E|C| = \E|A| \pm O(\epsilon n) = \E|\GMM{G, \pi}| \pm O(\epsilon n)$. Subtracting a sufficiently large additive factor of $O(\epsilon n)$ from the output gives \cref{prop:sublinear}.
\end{proof}


\vspace{-0.5cm}

\begin{proof}[Proof of \cref{prop:given-vertex}]
    It is shown in \cite[Eq (9)]{behnezhad2021} that for any arbitrary vertex $u \in V_H$, one can determine which edge of $u$ (if any) belongs to \GMM{H, \pi} in time $T(u, \pi)$ where $\sum_{u \in V_H} \E_\pi[T(u, \pi)] = \widetilde{O}(n^2 / \epsilon)$. Thus, for $v \sim K$ chosen u.a.r., 
    $
    \E_{\pi, v}[T(v, \pi)] = \frac{1}{|K|} \sum_{v \in K} \E_\pi[T(v, \pi)] = O(n^2/(\epsilon|K|)),
    $ which is the claimed time. Since the edge of $v$ in $\GMM{H, \pi}$ is also given (if any) by this process, we can determine if it belongs to $C$. By \cref{cl:match-statusAC}, at most $3\epsilon n$ vertices in $V_1$ (and so in $K$) have a different match-status in $C$ and $A = \GMM{G, \pi}$. Thus, probability (taken over $v \sim K$) of choosing a vertex that is not among these $\epsilon n$ out of $K$ is at least $(|K| - 3\epsilon n)/|K| = 1 - O(\epsilon n /|K|)$.
\end{proof}

\section{Proof of \cref{lem:A-to-B}}\label{sec:A-to-B-proof}

In this section, we argue why \cref{lem:A-to-B} holds. Let us restate the lemma first.

\noindent \textbf{\cref{lem:A-to-B} (restated).} {\em \lemAtoBstatement{}}

Let us prove the lemma step by step. First, we prove the following lemma which guarantees all the desired properties of algorithm $\mc{B}$, except that instead of a multiplicative approximation, it achieves a multiplicative-additive $(\alpha, \epsilon n)$-approximation.

\begin{lemma}\label{lem:xggg-t210}
    Let algorithm $\mc{A}$ be as in \cref{lem:A-to-B}. There is an algorithm $\mc{C}$ that maintains a number $\widetilde{\mu}_c$ such that at any point during the updates, it holds w.h.p. that $\alpha \mu(G) - O(\epsilon n) \leq \widetilde{\mu}_c \leq \mu(G)$. Algorithm $\mc{C}$ takes $O\left(\left(U(n) + \frac{Q(n, \epsilon)}{n} \right) \poly(\log n, 1/\epsilon) \right)$ worst-case update-time. Moreover, if $\mc{A}$ works against adaptive adversaries, then so does $\mc{C}$.
\end{lemma}
\begin{myproof}
    We run algorithm $\mc{A}$ in the background, paying a worst-case update-time of $U(n)$ because of it. We then take the lazy approach. We query $\mc{A}$ to produce the estimate $\widetilde{\mu}$, return $\widetilde{\mu}' := \widetilde{\mu} - 2\epsilon n$ as our output, then we do not change the output for the next $\epsilon n$ updates, and repeat the same process. Because every edge update can change the size of the maximum matching by at most one (even against adaptive adversaries), it will hold at all times that $\alpha \mu(G) - 3\epsilon n \leq \E[\widetilde{\mu}'] \leq \mu(G) - \epsilon n$. Moreover, because we query algorithm $\mc{A}$ every $\epsilon n$ updates, the overall amortized update time of the algorithm is $O(U(n) + \frac{Q(n, \epsilon)}{\epsilon n})$. It remains to $(i)$ turn the amortized update-time bound to worst-case, and  $(ii)$ turn the expected approximation bound to a high probability bound.
    
    Let us address $(i)$ first. This can be done using a well-known `spreading' idea (see \cite{GuptaPeng-FOCS13} for more details). Instead of executing the oracle of algorithm $\mc{A}$ over one update, we spread it over multiple updates. More precisely, we spread the $Q(n, \epsilon)$ time needed for the oracle over $\epsilon n / 2$ updates, each performing $2\frac{Q(n, \epsilon)}{\epsilon n}$ operations of it. When the process finishes, we update our solution as before, and immediately start spreading the next call to the oracle.
    
    We now address $(ii)$. To turn the approximation guarantee into a high probability bound, we simply run $O(\log n)$ independent instances of the algorithm above, and return the average of these $O(\log n)$ outputs as our output. Note that if algorithm $\mc{A}$ works against adaptive adversaries, then so should all of these $O(\log n)$ instances and thus our algorithm as well.
\end{myproof}

Let us now show how we can get rid of the additive $\epsilon n$ error.  First, note that if the maximum matching size is guaranteed to be $\Omega(n)$ at all times, then a multiplicative-additive $(\alpha, \epsilon n)$-approximation for it is indeed a multiplicative $(\alpha - O(\epsilon))$-approximation. Indeed up to a $\poly\log n$ increase in the update-time, this assumption comes w.l.o.g. due to a ``vertex sparsification'' idea of the literature \cite{AssadiKL16-Journal,Kiss-ITCS22}. In particular, suppose $\mu(G) \ll n$ and suppose that we know $\mu(G)$. Then the idea is to randomly contract the vertices into $O(\mu(G)/\epsilon)$ vertices, then remove self-loops and parallel edges. This way, it is not hard to see that the resulting graph still has a matching of size at least $(1-\epsilon)\mu(G)$ in expectation, but has much fewer vertices. The problem with this approach is that it only works against oblivious adversaries since an adaptive adversary may insert edges among the contracted nodes. However, \cite{Kiss-ITCS22} showed that taking $\poly\log n$ of these vertex sparsified subgraphs are resilient against adaptive adversaries also. In particular, the following was proved in \cite{Kiss-ITCS22}, which combined with \cref{lem:xggg-t210} implies \cref{lem:A-to-B}.

\begin{proposition}[{\cite[Corollary~4.101]{Kiss-ITCS22}}]
    If there is a dynamic algorithm for maintaining an $(\alpha, \delta n)$-approximate maximum matching for dynamic graphs in update time $O(T(n, \delta))$ then there is a randomized algorithm for maintaining an $(\alpha - \epsilon)$-approximate maximum matching with update time $O(T(n, \epsilon) \log^4 n/\epsilon^8)$ which works against adaptive adversaries given the underlying algorithm also does.
\end{proposition}

\section{Basic Facts}\label{sec:deferred}

\begin{fact}\label{fact:quadratic-sum}
    Let $x_1, \ldots, x_n$ be an arbitrary set of reals. Then denoting $\bar{x} = \frac{1}{n} \sum_{i=1}^n x_i$, it holds that $\sum_{i=1}^n x_i^2 = n \bar{x}^2 + \sum_{i=1}^n (x_i - \bar{x})^2 \geq n \bar{x}^2$.
\end{fact}
\begin{myproof}
    For any $i \in [n]$ denote $y_i = (x_i - \bar{x})$. We have
    $$
        \sum_{i=1}^n x_i^2 = \sum_{i=1}^n (\bar{x} + y_i)^2 = \sum_{i=1}^n (\bar{x}^2 + y_i^2 + 2\bar{x} y_i) = n \bar{x}^2 + \sum_{i=1}^n y_i^2 + 2\bar{x} \sum_{i=1}^n y_i = n \bar{x}^2 + \sum_{i=1}^n y_i^2,
    $$
    where the last equality follows from $\sum_{i=1}^n y_i = \sum_{i=1}^n (x_i - \bar{x}) = \sum_{i=1}^n x_i - \sum_{i=1}^n \bar{x} = 0$.
\end{myproof}

\end{document}